\documentclass{article}
\usepackage{microtype}
\usepackage{graphicx}
\usepackage{subfigure}
\usepackage{booktabs} % for professional tables

\usepackage{hyperref}

\usepackage{times}
\usepackage{graphicx} % more modern
\usepackage{subfigure}

\usepackage{natbib}

\usepackage{algorithm}
\usepackage{algorithmic}

\usepackage{amsmath}
\usepackage{amssymb}

\usepackage{amsthm}
\usepackage{wrapfig}
\usepackage{multicol}

\usepackage{hyperref}

\makeatletter
\newcommand{\rmnum}[1]{\romannumeral #1}
\newcommand{\Rmnum}[1]{\expandafter\@slowromancap\romannumeral #1@}
\makeatother

\newtheorem{theorem}{\textbf{Theorem}}
\newtheorem{lemma}{\textbf{Lemma}}

\newtheorem{remark}{\textbf{Remark}}

\newtheorem{claim}{Claim}

\usepackage{lipsum}
\usepackage{amsfonts}
\usepackage{tabularx}

\usepackage{helvet}
\usepackage{courier}
\usepackage{ctable}
\usepackage{color}
\usepackage{lipsum}
\usepackage{multicol, multirow}

\newcolumntype{L}{>{\normalsize}l}
\newcolumntype{C}{>{\normalsize}c}

\usepackage{microtype}
\usepackage{graphicx}
\usepackage{subfigure}
\usepackage{booktabs} % for professional tables

\usepackage{hyperref}

\usepackage[accepted]{icml2020}

\icmltitlerunning{Layered Sampling for Robust Optimization Problems}

\begin{document}

\twocolumn[
\icmltitle{Layered Sampling for Robust Optimization Problems}



%\icmlsetsymbol{equal}{}

\begin{icmlauthorlist}
\icmlauthor{Hu Ding}{to}
\icmlauthor{Zixiu Wang}{to}
%\icmlauthor{Cieua Vvvvv}{to}
%\icmlauthor{Iaesut Saoeu}{to}
%\icmlauthor{Fiuea Rrrr}{to}
%\icmlauthor{Tateu H.~Yasehe}{to}
%\icmlauthor{Aaoeu Iasoh}{to}
%\icmlauthor{Buiui Eueu}{ed}
%\icmlauthor{Aeuia Zzzz}{ed}
%\icmlauthor{Bieea C.~Yyyy}{to,goo}
%\icmlauthor{Teoau Xxxx}{ed}
%\icmlauthor{Eee Pppp}{ed}
\end{icmlauthorlist}

\icmlaffiliation{to}{School of Computer Science and Technology, University of Science and Technology of China}
%\icmlaffiliation{goo}{Googol ShallowMind, New London, Michigan, USA}
%\icmlaffiliation{ed}{School of Computation, University of Edenborrow, Edenborrow, United Kingdom}

\icmlcorrespondingauthor{Hu Ding}{huding@ustc.edu.cn, \url{http://staff.ustc.edu.cn/~huding/} }
%\icmlcorrespondingauthor{Eee Pppp}{ep@eden.co.uk}

\icmlkeywords{Machine Learning, ICML}

\vskip 0.3in
]



\printAffiliationsAndNotice{}  % leave blank if no need to mention equal contribution

\begin{abstract}
In real world,  our datasets often contain outliers. 
Moreover, the outliers can seriously affect the final machine learning result. Most existing algorithms for handling outliers take high time complexities ({\em e.g.} quadratic or cubic complexity).
{\em Coreset} is a popular approach for compressing data so as to speed up the optimization algorithms. However, the current coreset methods cannot be easily extended to handle the case with outliers. 
In this paper, we propose a new variant of coreset technique,  {\em layered sampling}, to deal with two fundamental robust optimization problems: {\em $k$-median/means clustering with outliers} and {\em linear regression with outliers}. This new coreset method is in particular suitable to speed up the iterative algorithms (which often improve the solution within a local range) for those robust optimization problems. Moreover, our method is easy to be implemented in practice. We expect that our framework of layered sampling will be applicable to  other robust optimization problems.
 \end{abstract}

\section{Introduction}
\label{sec-intro}

{\em Coreset} is a widely studied technique for solving many optimization problems~\cite{DBLP:journals/corr/Phillips16,bachem2017practical,munteanu2018coresets,DBLP:journals/widm/Feldman20}. The (informal) definition is as follows. Given an optimization problem with the objective function $\Delta$, denote by $\Delta(P, C)$ the objective value determined by a dataset $P$ and a solution $C$; a small set $S$ is called a coreset if 
\begin{eqnarray}
\Delta(P, C)\approx \Delta(S, C) \label{for-gcoreset}
\end{eqnarray}
for any feasible solution $C$. Roughly speaking, the coreset is a small set of data approximately representing a much larger dataset, and therefore existing algorithm can run on the coreset (instead of the original dataset) so as to reduce the complexity measures like running time, space, and communication. In the past years, the coreset techniques have been successfully applied to solve many optimization problems, such as clustering~\cite{chen2009coresets,DBLP:conf/stoc/FeldmanL11,huang2018epsilon}, logistic regression~\cite{huggins2016coresets,munteanu2018coresets}, linear regression~\cite{dasgupta2009sampling,drineas2006sampling}, and Gaussian mixture model~\cite{lucic2017training,DBLP:journals/corr/abs-1906-04845}. 

A large part of existing coreset construction methods are based on the theory of {\em sensitivity} which was proposed by \cite{langberg2010universal}. Informally, each data point $p\in P$ has the sensitivity $\phi(p)$ (in fact, we just need to compute an appropriate upper bound of $\phi(p)$) to measure its importance to the whole instance $P$ over all possible solutions, and $\Phi(P)=\sum_{p\in P} \phi(p)$ is called the total sensitivity. The coreset construction is a simple sampling procedure where each point $p$ is drawn {\em i.i.d.} from $P$ proportional to $\frac{\phi(p)}{\Phi(P)}$; each sampled point $p$ is assigned a weight $w(p)=\frac{\Phi(P)}{m \phi(p)}$ where $m$ is the sample size depending on the ``pseudo-dimension'' of the objective function $\Delta$ (\cite{DBLP:conf/stoc/FeldmanL11,li2001improved}); eventually, the set of weighted sampled points form the desired coreset $S$.

In real world, datasets are noisy and contain outliers. Moreover, outliers could seriously affect the final results in data analysis~\cite{chandola2009anomaly,DBLP:journals/cacm/GoodfellowMP18}. However, the sensitivity based coreset approach is not appropriate to handle robust optimization problems involving outliers ({\em e.g.}, $k$-means clustering with outliers). For example, it is not easy to compute the sensitivity $\phi(p)$ because the point $p$ could be inlier or outlier for different solutions; moreover, it is challenging to build the relation, such as (\ref{for-gcoreset}), between the original instance $P$ and the coreset $S$ ({\em e.g.}, how to determine the number of outliers for the instance $S$?).

\subsection{Our Contributions} 
In this paper, we consider two important robust optimization problems: {\em $k$-median/means clustering with outliers} and {\em linear regression with outliers}. Their quality guaranteed algorithms exist but  often have high complexities that seriously limit their applications in real scenarios (see Section~\ref{sec-relate} for more details). We observe that these problems can be often efficiently solved by some heuristic algorithms in practice, though they only guarantee local optimums in theory. For example, \cite{chawla2013k} proposed the algorithm {\em $k$-means-~-} to solve the problem of $k$-means clustering with outliers, where the main idea is an {\em alternating minimization} strategy. The algorithm is an iterative procedure, where it alternatively updates the outliers and the $k$ cluster centers in each iteration; eventually the solution converges to a local optimum. The alternating minimization strategy is also widely used for solving the problem of linear regression with outliers, {\em e.g.,}~\cite{ShenS19}. A common feature of these methods is that they usually start from an initial solution and then locally improve the solution round by round. Therefore, a natural question is 

{\em can we construct a ``coreset'' only for a local range in the solution space?}

Using such a coreset, we can substantially speed up those iterative algorithms. 
Motivated by this question, we 
introduce a new variant of coreset method called {\em layered sampling}. Given an initial solution $\tilde{C}$, we partition the given data set $P$ into a consecutive sequence of ``layers'' surrounding $\tilde{C}$ and conduct the random sampling in each layer; the union of the samples, together with the points located in the outermost layer, form the coreset $S$. Actually, our method is partly inspired by the coreset construction method of $k$-median/means clustering (without outliers) proposed by \cite{chen2009coresets}. However, we need to develop significantly new idea in theory to prove its correctness for the case with outliers. The purpose of layered sampling is not to guarantee the approximation quality (as (\ref{for-gcoreset})) for any solution $C$, instead, it only guarantees the quality for the solutions in a local range $\mathcal{L}$ in the solution space (the formal definition is given in Section~\ref{sec-pre}). Informally, we need to prove the following result to replace (\ref{for-gcoreset}): 
\begin{eqnarray}
\forall C\in\mathcal{L}, \Delta(P, C)\approx \Delta(S, C) \label{for-gcoreset2}
\end{eqnarray}
See Figure~\ref{fig-local} for an illustration. In other words,  the new method can help us to find a local optimum faster. Our main results are shown in Theorem~\ref{the-main} and \ref{the-main-lr}. The construction algorithms are easy to implement. 
%\textbf{Due to the space limit, we leave the detailed experimental results to our supplement.}

\begin{figure}[]
    \centering
  \includegraphics[height=0.8in]{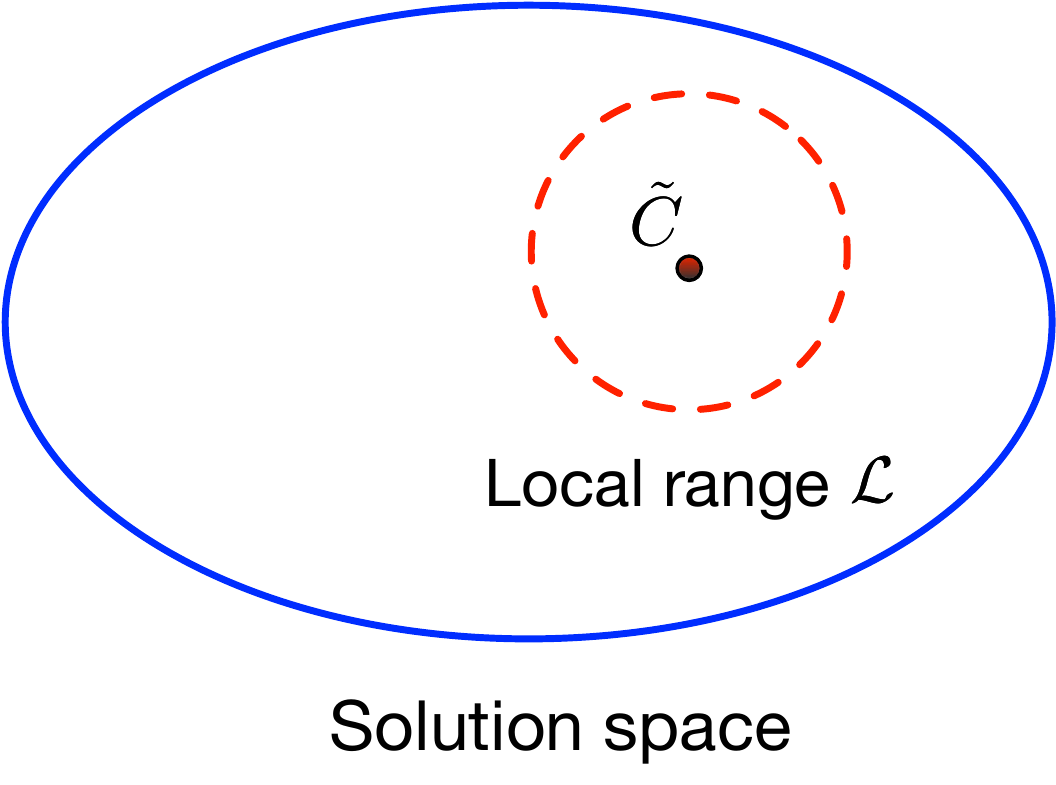}
  \vspace{-0.1in}
      \caption{The red point represents the initial solution $\tilde{C}$, and our goal is to guarantee (\ref{for-gcoreset2}) for a local range around $\tilde{C}$.}
  \label{fig-local}
  \vspace{-0.1in}
\end{figure}

  \vspace{-0.05in}
\subsection{Related Works}
\label{sec-relate}
  \vspace{-0.05in}

\textbf{$k$-median/means clustering (with outliers).}  $k$-median/means clustering are two popular center-based clustering problems~\cite{awasthi2014center}. 
It has been extensively studied for using coreset techniques to reduce the complexities of $k$-median/means clustering algorithms~\cite{chen2009coresets,har2007smaller,fichtenberger2013bico,DBLP:conf/soda/FeldmanSS13}; in particular,  \cite{DBLP:conf/stoc/FeldmanL11} proposed a unified coreset framework for a set of clustering problems. However, the research on using coreset to handle outliers is still quite limited. Recently,  \cite{huang2018epsilon} showed that a uniform independent sample can serve as a coreset for clustering with outliers in Euclidean space; however, such uniform sampling based method often misses some important points and therefore introduces an unavoidable error on the number of outliers. 
\cite{gupta2018approximation} also studied the uniform random sampling idea but under the assumption that  
each optimal cluster should be large enough. Partly inspired by the method of~\cite{mettu2004optimal}, \cite{DBLP:conf/nips/ChenA018} proposed a novel summary construction algorithm to reduce input data size which guarantees an $O(1)$ factor of distortion on the clustering cost.

In theory, the algorithms with provable guarantees for $k$-median/means clustering with outliers ~\cite{chen2008constant,krishnaswamy2018constant,friggstad2018approximation} have high complexities and are difficult to be implemented in practice. The heuristic but practical algorithms  have also been studied before~\cite{chawla2013k,ott2014integrated}. By using the local search method, \cite{gupta2017local} provided a $274$-approximation algorithm of $k$-means clustering with outliers but needing to discard more than the desired number of outliers; to improve the running time, they also used $k$-means++~\cite{arthur2007k} to seed the ``coreset'' that yields an $O(1)$ factor approximation. Based on the idea of $k$-means++, \cite{bhaskara2019greedy} proposed an $O(\log k)$-approximation algorithm.

\textbf{Linear regression (with outliers).} 
Several coreset methods for ordinary linear regression (without outliers) have been proposed~\cite{drineas2006sampling,dasgupta2009sampling,BoutsidisDM13}.
For the case with outliers, which is also called ``Least Trimmed Squares linear estimator (LTS)'', a uniform sampling approach was studied by \cite{DBLP:journals/algorithmica/MountNPSW14,Ding-2014}. But similar to the scenario of clustering with outliers, such uniform sampling approach introduces an unavoidable error on the number of outliers. 

\cite{DBLP:journals/algorithmica/MountNPSW14} also proved that it is impossible to achieve even an approximate solution for LTS within polynomial time under the conjecture of {\em the hardness of affine degeneracy}~\cite{DBLP:journals/dcg/EricksonS95}, if the dimensionality $d$ is not fixed. 
Despite of its high complexity, several practical algorithms were proposed before and most of them are based on the idea of alternating minimization that improves the solution within a local range, such as~\cite{Rousseeuw-1984,RousseeuwD06,Dougla-1994,DBLP:journals/csda/MountNPWS16,DBLP:conf/nips/BhatiaJK15,ShenS19}. \cite{DBLP:conf/colt/KlivansKM18} provided another approach based on the {\em sum-of-squares} method.

  \vspace{-0.06in}
\subsection{Preliminaries}
\label{sec-pre}
  \vspace{-0.06in}

Below, we introduce several important definitions.

\textbf{\rmnum{1}. $k$-Median/Means Clustering with Outliers.}  Suppose $P$ is a set of $n$ points in $\mathbb{R}^d$. Given two integers $1\leq z, k< n$, the problem of $k$-median clustering with $z$ outliers is to find a set of $k$ points $C=\{c_1, \cdots, c_k\}\subset\mathbb{R}^d$ and a subset $P'\subset P$ with $|P'|=n-z$, such that the following objective function
\begin{eqnarray}
\mathcal{K}_{1}^{-z} (P, C)=\frac{1}{n-z}\sum_{p\in P'}\min_{1\leq j\leq k}||p-c_j||\label{for-obj-k-med}
\end{eqnarray}
is minimized. Similarly, we have the objective function
\begin{eqnarray}
\mathcal{K}_{2}^{-z} (P, C)= \frac{1}{n-z}\sum_{p\in P'}\min_{1\leq j\leq k}||p-c_j||^2, \label{for-obj-k-mea}
\end{eqnarray}
for $k$-means clustering with outliers. The set $C$ is also called a solution of the instance $P$. Roughly speaking, given a solution $C$, the farthest $z$ points to $C$ are discarded, and the remaining subset $P'$ is partitioned into $k$ clusters where each point is assigned to its nearest neighbor of $C$.

\textbf{\rmnum{2}. Linear Regression with Outliers.} Given a vector $h=(h_1, h_2, \cdots, h_d)\in\mathbb{R}^d$, the linear function defined by $h$ is $y=\sum^{d-1}_{j=1}h_j x_{j}+h_d$ for $d-1$ real variables $x_1, x_2, \cdots, x_{d-1}$. Thus the linear function can be represented by the vector $h$. From geometric perspective, the linear function can be viewed as a $(d-1)$-dimensional hyperplane in the space.  Let $z$ be an integer between $1$ and $n$, and $P=\{p_1, p_2, \cdots, p_n\}$ be a set of $n$ points in $\mathbb{R}^d$, where each $p_i=(x_{i, 1}, x_{i, 2}, \cdots, x_{i, d-1}, y_i)$ for $1\leq i\leq n$; the objective is to find a subset $P'\subset P$ with $|P'|=n-z$ and a $(d-1)$-dimensional hyperplane, represented as a coefficient vector $h=(h_1, h_2, \cdots, h_d)\in\mathbb{R}^d$, such that 
\begin{eqnarray}
\mathcal{LR}_{1}^{-z} (P', h)=& \frac{1}{n-z}\sum_{p_i\in P'}\big |Res(p_i, h)\big| \label{for-obj-lr1}\\
\text{ or }\hspace{0.2in}\mathcal{LR}_{2}^{-z} (P', h)=&\frac{1}{n-z}\sum_{p_i\in P'}\big(Res(p_i, h)\big)^2 \label{for-obj-lr2}
\end{eqnarray}
is minimized. $Res(p_i, h)=y_i-\sum^{d-1}_{j=1}h_j x_{i,j}-h_d$ is the ``residual'' of $p_i$  to $h$.  The objective functions (\ref{for-obj-lr1}) and (\ref{for-obj-lr2}) are called the ``least absolute error''  and ``least squared error'', respectively.

\begin{remark}
\label{re-weight}
All the above problems can be extended to weighted case. Suppose each point $p$ has a non-negative weight $w(p)$, then the (squared) distance $||p-c_j||$ ($||p-c_j||^2$) is replaced by $w(p)\cdot||p-c_j||$ ($w(p)\cdot||p-c_j||^2$); we can perform the similar modification on $\big |Res(p_i, h)\big|$ and $\big(Res(p_i, h)\big)^2$ for the problem of linear regression with outliers. Moreover, the total weights of the outliers should be equal to $z$. Namely, we can view each point $p$ as $w(p)$ unit-weight overlapping points. 
\end{remark}

\textbf{Solution range.} To analyze the performance of our layered sampling method, we also need to define the ``solution range'' for the clustering  and linear regression problems. Consider the clustering problems  first. Given a clustering solution $\tilde{C}=\{\tilde{c}_1, \cdots, \tilde{c}_k\}\subset\mathbb{R}^d$ and $L>0$, we use ``$\tilde{C}\pm L$'' to denote the range of solutions $\mathcal{L}=$
\begin{eqnarray}
\Big\{C=\{c_1, \cdots, c_k\}\mid ||\tilde{c}_j-c_j||\leq L, \forall 1\leq j\leq k\Big\}. \label{for-range-clustering}
\end{eqnarray}

Next, we define the solution range for linear regression with outliers. Given an instance $P$, we often normalize the values in each of the first $d-1$ dimensions as the preprocessing step; without loss of generality, we can assume that $x_{i, j}\in [0, D]$ with some $D>0$ for any $1\leq i\leq n$ and $1\leq j\leq d-1$. For convenience, denote by $\mathcal{R}_D$ the region $\{(s_1, s_2, \cdots, s_d)\mid 0\leq s_j\leq D, \forall 1\leq j\leq d-1\} $ and thus $P\subset \mathcal{R}_D$ after the normalization. It is easy to see that the region $\mathcal{R}_D$ actually is a vertical square cylinder in the space. Given a coefficient vector (hyperplane) $\tilde{h}=(\tilde{h}_1, \tilde{h}_2, \cdots, \tilde{h}_d)\in\mathbb{R}^d$ and $L>0$, we use ``$\tilde{h}\pm L$'' to denote the range of hyperplanes $\mathcal{L}=$
\begin{eqnarray}
&&\Big\{h=(h_1, h_2, \cdots, h_d)\mid \nonumber\\
&&|Res(p, \tilde{h})-Res(p, h)|\leq L, \forall p\in  \mathcal{R}_D\Big\}. \label{for-range-regression}
\end{eqnarray}

To understand the range defined in (\ref{for-range-regression}), we can imagine two linear functions $\tilde{h}^+=(\tilde{h}_1, \tilde{h}_2, \cdots, \tilde{h}_d+L)$ and $\tilde{h}^-=(\tilde{h}_1, \tilde{h}_2, \cdots, \tilde{h}_d-L)$; if we only consider the region $ \mathcal{R}_D$, the range (\ref{for-range-regression}) contains all the linear functions ``sandwiched'' by $\tilde{h}^+$ and $\tilde{h}^-$.

For both the clustering and regression problems, we also say that \textbf{the size of the solution range} $\mathcal{L}$ is $|\mathcal{L}|=L$.

\section{The Layered Sampling Framework}
\label{sec-framework}
We present the overview of our layered sampling framework. 
 For the sake of completeness, we first introduce the coreset construction method for the ordinary $k$-median/means clustering proposed by~\cite{chen2009coresets}. 
 
Suppose $\alpha$ and $\beta\geq 1$. A ``bi-criteria $(\alpha, \beta)$-approximation'' means that it contains $\alpha k$ cluster centers, and the induced clustering cost is at most $\beta$ times the optimum. Usually, finding a bi-criteria approximation is much easier than achieving a single-criterion approximation. For example, one can obtain a bi-criteria approximation for $k$-median/means clustering in linear time with $\alpha=O(1)$ and $\beta=O(1)$~\cite{chen2009coresets}. Let $T=\{t_1, t_2, \cdots, t_{\alpha k}\}\subset\mathbb{R}^d$ be the obtained $(\alpha, \beta)$-approximate solution of the input instance $P$. For convenience, we use $\mathbb{B}(c, r)$ to denote the ball centered at a point $c$ with radius $r>0$. 
At the beginning of  Chen's coreset construction algorithm, it takes two carefully designed values $r>0$ and $N=O(\log n)$, and partitions the space into $N+1$ layers $H_0, H_1, \cdots, H_N$, where $H_0=\cup^{\alpha k}_{j=1}\mathbb{B}(t_j, r)$ and $H_{i}=\big(\cup^{\alpha k}_{j=1}\mathbb{B}(t_j, 2^i r)\big)\setminus\big(\cup^{\alpha k}_{j=1}\mathbb{B}(t_j, 2^{i-1}r)\big)$ for $1\leq i\leq N$. It can be proved that $P$ is covered by $\cup^N_{i=0}H_i$; then the algorithm takes a random sample $S_i$ from each layer $P\cap H_i$, and the union $\cup^N_{i=0}S_i$ forms the desired coreset $S$ satisfying the condition (\ref{for-gcoreset}). 

However, this approach cannot directly solve the case with outliers. First, it is not easy to obtain a bi-criteria approximation for the problem of $k$-median/means clustering with outliers ({\em e.g.,} in linear time). Moreover, it is challenging to guarantee the condition (\ref{for-gcoreset}) for any feasible solution $C$, because the set of outliers could change when $C$ changes (this is also the major challenge for proving the correctness of our method later on). We propose a modified version of Chen's coreset construction method and  aim to guarantee (\ref{for-gcoreset2}) for a local range of solutions. We take the $k$-median clustering with outliers problem as an example. Let $\tilde{C}=\{\tilde{c}_1, \cdots, \tilde{c}_k\}\subset\mathbb{R}^d$ be a given solution. Assume $\epsilon>0$ and $N\in \mathbb{Z}^+$ are two pre-specified parameters. With a slight abuse of notations, we still use $H_0, H_1, \cdots, H_N$ to denote the layers surrounding $\tilde{C}$, {\em i.e.,}  
\begin{eqnarray}
H_0&=&\cup^{k}_{j=1}\mathbb{B}(\tilde{c}_j, r); \label{for-frame5}\\
H_{i}&=&\big(\cup^{ k}_{j=1}\mathbb{B}(\tilde{c}_j, 2^i r)\big)\setminus\big(\cup^{ k}_{j=1}\mathbb{B}(\tilde{c}_j, 2^{i-1}r)\big) \nonumber\\
&&\text{ for }1\leq i\leq N. \label{for-frame6}
\end{eqnarray}
In addition, let
\begin{eqnarray}
H_{out}=\mathbb{R}^d\setminus \big(\cup^{ k}_{j=1}\mathbb{B}(\tilde{c}_j, 2^N r)\big).\label{for-frame7}
\end{eqnarray}
Here, we set the value $r$ to satisfy the following condition:
\begin{eqnarray}
\Big|P\cap H_{out}\Big|=(1+\frac{1}{\epsilon})z. \label{for-frame1}
\end{eqnarray}
That is, the union of the layers $\cup^N_{i=0}H_i$ covers $n-(1+\frac{1}{\epsilon})z$ points of $P$ and excludes the farthest $(1+\frac{1}{\epsilon})z$. Obviously, such a value $r$ always exists. Suppose $P'$ is the set of $n-z$ inliers induced by $\tilde{C}$, and then we have 
\begin{eqnarray}
2^N r&\leq&\frac{\epsilon}{z}\sum_{p\in P'}\min_{1\leq j\leq k}||p-\tilde{c}_j||\nonumber\\
&=& \frac{\epsilon}{z}(n-z)\mathcal{K}^{-z}_1(P, \tilde{C}) \label{for-r}
\end{eqnarray}
via the Markov's inequality. Our new coreset contains the following $N+2$ parts:
\begin{eqnarray}
S=S_0\cup S_1\cup\cdots\cup S_N\cup S_{out},\label{for-frame2}
\end{eqnarray}
where $S_i$ is still a random sample from $P\cap H_i$ for $0\leq i\leq N$, and $S_{out}$ contains all the $(1+\frac{1}{\epsilon})z$ points in $H_{out}$. In Section~\ref{sec-construction}, we will show that the coreset $S$ of (\ref{for-frame2}) satisfies (\ref{for-gcoreset2}) for the $k$-median clustering with outliers problem (and similarly for the $k$-means clustering with outliers problem).

For the linear regression with outliers problem, we apply the similar layered sampling framework. Define $\mathbb{S}(h, r)$ to be the slab centered at a $(d-1)$-dimensional hyperplane $h$ with $r>0$, {\em i.e.,} $\mathbb{S}(h, r)=\{p\in \mathbb{R}^d\mid -r\leq Res(p, h)\leq r\}$.  Let $P$ be an instance, and $\tilde{h}=(\tilde{h}_1, \cdots, \tilde{h}_d)\in\mathbb{R}^d$ be a given hyperplane. We divide the space into $N+2$ layers $H_0, H_1, \cdots, H_N, H_{out}$, where 
\begin{eqnarray}
H_0&=&\mathbb{S}(\tilde{h}, r); \label{for-frame8}\\
H_i&=&\mathbb{S}(\tilde{h}, 2^i r)\setminus \mathbb{S}(\tilde{h}, 2^{i-1}r) \text{ for } 1\leq i\leq N;\label{for-frame9}\\
 H_{out}&=&\mathbb{R}^d\setminus \mathbb{S}(\tilde{h}, 2^{N}r). \label{for-frame10}
 \end{eqnarray} 
Similar to (\ref{for-frame1}), we also require the value $r$ to satisfy the following condition:
\begin{eqnarray}
\Big|P\cap H_{out}\Big|=\Big|P\setminus \mathbb{S}(\tilde{h}_N, 2^N r)\Big|=(1+\frac{1}{\epsilon})z. \label{for-frame3}
\end{eqnarray}
And consequently, we have
\begin{eqnarray}
2^N r\leq \frac{\epsilon}{z}(n-z)\mathcal{LR}^{-z}_1(P, \tilde{h}). \label{for-frame4}
\end{eqnarray}
Then, we construct the coreset for linear regression with outliers by the same manner of (\ref{for-frame2}).

\section{$k$-Median/Means Clustering with Outliers}
\label{sec-construction}
In this section, we provide the details on applying our layered sampling framework to the problem of $k$-median clustering with outliers. 
See Algorithm~\ref{alg-core}. 
The algorithm and analysis can be easily modified to handle $k$-means clustering with outliers, where the only difference is that we need to replace (\ref{for-r}) by `` $2^N r\leq \sqrt{\frac{\epsilon}{z}(n-z)\mathcal{K}^{-z}_2(P, \tilde{C})}$ ''.

\begin{algorithm}[tb]
   \caption{\sc{Layered Sampling for $k$-Med-Outlier}}
   \label{alg-core}
\begin{algorithmic}
  \STATE {\bfseries Input:} An instance $P\subset\mathbb{R}^d$ of $k$-median clustering with $z$ outliers, a solution $\tilde{C}=\{\tilde{c}_1, \cdots, \tilde{c}_k\}$, and two parameters $\epsilon, \eta\in (0,1)$.
   \STATE
\begin{enumerate}
\item Let $\gamma=z/(n-z)$ and $N=\lceil\log \frac{1}{\gamma}\rceil$. Compute the value $r$ satisfying (\ref{for-frame1}). 
\item As described in (\ref{for-frame5}), (\ref{for-frame6}), and (\ref{for-frame7}), the space is partitioned into $N+2$ layers $H_0, H_1, \cdots, H_N$ and $H_{out}$.
\item Randomly sample $\min\Big\{O(\frac{1}{\epsilon^2}kd\log\frac{d}{\epsilon}\log\frac{N}{\eta}), |P\cap H_i|\Big\}$ points, denoted by $S_{i}$, from $P\cap H_i$ for $0\leq i\leq N$.
\item For each point $p\in S_{i}$, set its weight to be $|P\cap H_i|/|S_i|$; let $S_H=\cup^N_{i=0}S_i$.
\end{enumerate}
 
  \STATE {\bfseries Output} $S=S_H\cup (P\cap H_{out})$.
\end{algorithmic}
\end{algorithm}

\begin{theorem}
\label{the-main}
Algorithm~\ref{alg-core} returns a point set $S$ having the size $|S|=\tilde{O}\footnote{The asymptotic notation $\tilde{O}(f)=O\big(f\cdot polylog(\frac{d}{\gamma\epsilon\eta})\big)$.}(\frac{1}{\epsilon^2}kd)+(1+\frac{1}{\epsilon})z$. Moreover, with probability at least $1-\eta$, for any $L>0$ and any solution $C\in \tilde{C}\pm L$, we have 
\begin{eqnarray}
\mathcal{K}_{1}^{-z} (S, C) \in\mathcal{K}_{1}^{-z} (P, C)\pm\epsilon\big(\mathcal{K}_{1}^{-z} (P, \tilde{C})+L\big).\label{for-main}
\end{eqnarray}
Here, $S$ is a weighted instance of $k$-median clustering with outliers, and the total weight of outliers is $z$ (see Remark~\ref{re-weight}).
 \end{theorem}
\begin{remark}
(1) The running time of Algorithm~\ref{alg-core} is $O(knd)$. For each point $p\in P$, we compute its shortest distance to $\tilde{C}$, $\min_{1\leq j\leq k}||p-\tilde{c}_j||$; then select the farthest $(1+1/\epsilon)z$ points and compute the value $r$ by running the linear time selection algorithm~\cite{blum1973time}; finally, we obtain the $N+1$ layers $H_i$ with $0\leq i\leq N$ and take the samples $S_0, S_1, \cdots, S_N$ from them. 

(2) Comparing with the standard coreset (\ref{for-gcoreset}), our result contains an additive error $\epsilon\big(\mathcal{K}_{1}^{-z} (P, \tilde{C})+L\big)$ in (\ref{for-main}) that depends on the initial objective value $\mathcal{K}_{1}^{-z} (P, \tilde{C})$ and the size $L$ of the solution range. In particular, the smaller the range size $L$, the lower the error of our coreset. 

(3) The algorithm of \cite{DBLP:conf/nips/ChenA018} also returns a summary for compressing the input data. But there are two major differences comparing with our result. First, their summary guarantees a constant factor of distortion on the clustering cost, while our error approaches $0$ if $\epsilon$ is small enough. Second, their construction algorithm (called ``successive sampling'' from~\cite{mettu2004optimal}) needs to scan the data multiple passes, while our Algorithm~\ref{alg-core} is much simpler and only needs to read the data in one-pass. We also compare these two methods in our experiments.
\end{remark}

To prove Theorem~\ref{the-main}, we first show that $S_H$ is a good approximation of $P\setminus H_{out}$. Fixing a solution $C\in \tilde{C}\pm L$, we view the distance from each point $p\in P $ to $C$, {\em i.e.}, $\min_{1\leq j\leq k}||p-c_j||$, as a random variable $x_p$. For any point $p\in P\cap H_i$ with $0\leq i\leq N$, we have the following bounds for $x_p$. Suppose $p$ is covered by $ \mathbb{B}(\tilde{c}_{j_1}, 2^i r)$. Let the nearest neighbor of $p$ in $C$ be $c_{j_2}$. Then, we have the upper bound
\begin{eqnarray}
x_p&=&||p-c_{j_2}||\leq ||p-c_{j_1}||\nonumber\\
&\leq& ||p-\tilde{c}_{j_1}||+||\tilde{c}_{j_1}-c_{j_1}||\leq 2^i r+L. \label{for-vbound}
\end{eqnarray}
Similarly, we have the lower bound
\begin{equation}
 \left.\begin{aligned}
       x_p&\geq\max\{ 2^{i-1}r-L, 0\}  &\text{if $i\geq 1$; }\\
         x_p&\geq 0  &\text{if $i=0$.}
                \end{aligned}
\hspace{0.3in} \right  \} \label{for-vboundlower}
  \qquad  
\end{equation}

Therefore, we can take a sufficiently large random sample $\hat{S}_{i}$ from $P\cap H_i$, such that $\frac{1}{|\hat{S}_{i}|}\sum_{p\in \hat{S}_{i}}x_p\approx \frac{1}{|P\cap H_i|}\sum_{p\in P\cap H_i}x_p$ with certain probability. Specifically, combining  (\ref{for-vbound}) and (\ref{for-vboundlower}), we have the following lemma through the Hoeffding's inequality.
\begin{lemma}
\label{lem-sample1}
Let $\eta\in (0,1)$. If we randomly sample $O( \frac{1}{\epsilon^2}\log \frac{1}{\eta})$ points, denote by $\hat{S}_{i}$, from $P\cap H_i$, then with probability $1-\eta$, 
\begin{eqnarray}
\big|\frac{1}{|\hat{S}_{i}|}\sum_{p\in \hat{S}_{i}}x_p- \frac{1}{|P\cap H_i|}\sum_{p\in P\cap H_i}x_p\big|\leq  \epsilon (2^{i}r+2L). \nonumber
\end{eqnarray}
\end{lemma}

Lemma~\ref{lem-sample1} is only for a fixed solution $C$. To guarantee the result for any $C\in \tilde{C}\pm L$, we discretize the range $\tilde{C}\pm L$. 
Imagine that we build a grid inside each $\mathbb{B}(\tilde{c}_j, L)$ for $1\leq j\leq k$, where the grid side length is $\frac{\epsilon}{\sqrt{d}} L$. Denote by $G_j$ the set of grid points inside each $\mathbb{B}(\tilde{c}_j, L)$, and then $\mathcal{G}=G_1\times G_2\times\cdots\times G_k$ contains $O\Big(\big(\frac{2\sqrt{d}}{\epsilon}\big)^{kd}\Big)$ $k$-tuple points of $\tilde{C}\pm L$ in total. 
We increase the sample size in Lemma~\ref{lem-sample1} via replacing $\eta$ by $\frac{\eta}{N\cdot |\mathcal{G}|}$ in the sample size ``$O( \frac{1}{\epsilon^2}\log \frac{1}{\eta})$''. As a consequence, through taking the union bound for the success probability, we have the following result.

\begin{lemma}
\label{lem-sample2}
 $S_{i}$ is the sample obtained from $P\cap H_i$ in Step~3 of Algorithm~\ref{alg-core} for $0\leq i\leq N$.  Then with probability $1-\eta$,  
\begin{eqnarray}
\big|\frac{1}{|S_{i}|}\sum_{p\in S_{i}}x_p- \frac{1}{|P\cap H_i|}\sum_{p\in P\cap H_i}x_p\big|\leq  \epsilon (2^{i}r+2L). \nonumber
\end{eqnarray}
for each $i=\{0, 1, \cdots, N\}$ and  any $C\in \mathcal{G}$.
\end{lemma}

Next, we show that for any $C\in \tilde{C}\pm L$ (in particular the solutions in $\big(\tilde{C}\pm L\big)\setminus \mathcal{G}$), Lemma~\ref{lem-sample2} is true. For any solution $C=\{c_1, \cdots, c_k\}\in \tilde{C}\pm L$, let $C'=\{c'_1, \cdots, c'_k\}$ be its nearest neighbor in $\mathcal{G}$, {\em i.e.}, $c'_j$ is the grid point of the cell containing $c_j$ in $G_j$, for $1\leq j\leq k$. Also, denote by ${x'}_p$ the distance $\min_{1\leq j\leq k}||p-c'_j||$. Then we consider to bound the error $\big|\frac{1}{|S_{i}|}\sum_{p\in S_{i}}x_p- \frac{1}{|P\cap H_i|}\sum_{p\in P\cap H_i}x_p\big|$ through $C'$. By using the triangle inequality, we have
\begin{eqnarray}
&&\big|\frac{1}{|S_{i}|}\sum_{p\in S_{i}}x_p- \frac{1}{|P\cap H_i|}\sum_{p\in P\cap H_i}x_p\big|\label{for-triangle} \\
&\leq& \underbrace{\big|\frac{1}{|S_{i}|}\sum_{p\in S_{i}}x_p- \frac{1}{|S_{i}|}\sum_{p\in S_{i}}x'_p\big|}_{\text{(a)}}\nonumber\\
&+&\underbrace{\big|\frac{1}{|S_{i}|}\sum_{p\in S_{i}}x'_p- \frac{1}{|P\cap H_i|}\sum_{p\in P\cap H_i}x'_p\big|}_{\text{(b)}}\nonumber\\
&+&\underbrace{\big|\frac{1}{|P\cap H_i|}\sum_{p\in P\cap H_i}x'_p- \frac{1}{|P\cap H_i|}\sum_{p\in P\cap H_i}x_p\big|}_{\text{(c)}}. \nonumber
\end{eqnarray}
In (\ref{for-triangle}), the term (b) is bounded by Lemma~\ref{lem-sample2} since $C'\in \mathcal{G}$. To bound the terms (a) and (c), we study the difference $|x_p-x'_p|$ for each point $p$. Suppose the nearest neighbor of $p$ in $C$ ({\em resp.}, $C'$) is $c_{j_1}$ ({\em resp.}, $c'_{j_2}$). Then, 
\begin{eqnarray}
x_p&=&||p-c_{j_1}||\leq ||p-c_{j_2}||\nonumber\\
&\leq &||p-c'_{j_2}||+||c'_{j_2}-c_{j_2}||\nonumber\\
&\leq& ||p-c'_{j_2}||+\epsilon L=x'_p+\epsilon L,
\end{eqnarray}
where the last inequality comes from the fact that $c'_{j_2}$ and $c_{j_2}$ are in the same grid cell with side length $\frac{\epsilon}{\sqrt{d}}L$. Similarly, we have $x'_p\leq x_p+\epsilon L$. Overall, $|x_p-x'_p|\leq \epsilon L$. 
As a consequence, the terms (a) and (c) in (\ref{for-triangle}) are both bounded by $\epsilon L$. Overall, (\ref{for-triangle}) becomes 
\begin{eqnarray}
&&\big|\frac{1}{|S_{i}|}\sum_{p\in S_{i}}x_p- \frac{1}{|P\cap H_i|}\sum_{p\in P\cap H_i}x_p\big|\nonumber\\
&\leq&  O(\epsilon)(2^{i}r+L). \label{for-sample22}
\end{eqnarray}

For convenience, we use $P_H$ to denote the set $\cup^N_{i=0}(P\cap H_i)$. 

\begin{lemma}
\label{the-sample}
Let $S_0, S_1, \cdots, S_N$ be the samples obtained in Algorithm~\ref{alg-core}. Then, with probability $1-\eta$,
\begin{eqnarray}
&&\frac{1}{n-z}\big|\sum^N_{i=0}\frac{|P\cap H_i|}{|S_i|}\sum_{p\in S_i}x_p-\sum_{p\in P_H}x_p\big|\nonumber\\
&\leq& O(\epsilon)\Big(\mathcal{K}^{-z}_1(P, \tilde{C})+L\Big) \label{for-the-sample}
\end{eqnarray}
for any $C\in\tilde{C}\pm L$.
\end{lemma}
\begin{proof}
For convenience, let $Err_i=\big|\frac{1}{|S_{i}|}\sum_{p\in S_{i}}x_p- \frac{1}{|P\cap H_i|}\sum_{p\in P\cap H_i}x_p\big|$ for $0\leq i\leq N$. 
We directly have $Err_i\leq O(\epsilon)(2^{i}r+L)$ from (\ref{for-sample22}). Moreover, the left hand-side of (\ref{for-the-sample}) $=\frac{1}{n-z}\sum^N_{i=0} |P\cap H_i|\cdot Err_i$
\begin{eqnarray}
&\leq&\frac{O(\epsilon)}{n-z}\sum^N_{i=0} |P\cap H_i|\cdot (2^{i}r+L)\nonumber\\
&=&O(\epsilon)\cdot \sum^N_{i=0} \frac{|P\cap H_i|}{n-z}  2^{i}r+O(\epsilon) L. \label{for-the-sample-1}
\end{eqnarray}
It is easy to know that the  term $\sum^N_{i=0} \frac{|P\cap H_i|}{n-z}  2^{i}r$ of (\ref{for-the-sample-1}) is at most $\frac{1}{n-z}(|P\cap H_0|r+2\sum_{p\in P_H\setminus H_0}x_p)\leq r+2\mathcal{K}^{-z}_1(P, \tilde{C})$. Note we set $N=\lceil\log \frac{1}{\gamma}\rceil$ in Algorithm~\ref{alg-core}. Together with (\ref{for-r}), we know $r\leq \epsilon\mathcal{K}^{-z}_1(P, \tilde{C})$ and thus $\sum^N_{i=0} \frac{|P\cap H_i|}{n-z}  2^{i}r\leq O(1)\mathcal{K}^{-z}_1(P, \tilde{C})$. So (\ref{for-the-sample}) is true. 
\end{proof}

Below, we always assume that (\ref{for-the-sample}) is true 
and consider to prove (\ref{for-main}) of Theorem~\ref{the-main}. 
The set $P$ is partitioned into two parts: $P^C_{in}$ and $P^C_{out}$ by $C$, where $P^C_{out}$ is the $z$ farthest points to $C$ ({\em i.e.}, the outliers) and $P^C_{in}=P\setminus P^C_{out}$. Similarly, the coreset  $S$ is also partitioned into two parts $S^C_{in}$ and $S^C_{out}$ by $C$, where $S^C_{out}$ is the set of outliers with total weights $z$. In other words, we need to prove
\begin{eqnarray}
\sum_{p\in S^C_{in}} w(p) x_p\approx \sum_{p\in P^C_{in}}x_p.
\end{eqnarray}
Consider two cases: (\rmnum{1}) $ P_H\setminus P^C_{in}=\emptyset$ and (\rmnum{2}) $P_H\setminus P^C_{in}\neq\emptyset$. Intuitively, the case (\rmnum{1}) indicates that the set $P^C_{in}$ occupies the whole region $\cup^N_{i=0} H_i$; the case (\rmnum{2}) indicates that the region $\cup^N_{i=0} H_i$ contains some outliers from $P^C_{out}$. In the following subsections, we prove that (\ref{for-main}) holds for both cases. 
For ease of presentation, we use $w(U)$ to denote the total weight of a weighted point set $U$ (please be not confused with $|U|$, which is the number of points in $U$).

\vspace{-0.1in}
\subsection{Case (\rmnum{1}): $P_H\setminus P^C_{in}=\emptyset$}
\label{sec-case1}

We prove the following key lemma first. 
\begin{lemma}
\label{cla-main1}
If $P_H\setminus P^C_{in}=\emptyset$, $S^C_{in}=S_H\cup (P^C_{in}\setminus P_H)$ and $S^C_{out}=P^C_{out}$ (recall $S_H=\cup^N_{i=0}S_i$ from Algorithm~\ref{alg-core}).
\end{lemma}
\begin{proof}
First, the assumption $P_H\setminus P^C_{in}=\emptyset$ implies 
\begin{eqnarray}
P_H &\subset& P^C_{in};\label{for-cla-main1-2}\\
|P^C_{in}\setminus P_H|&=&|P^C_{in}|-|P_H|.\label{for-cla-main1-1}
\end{eqnarray}
In addition, since $S_H\subset P_H$, we have $S_H\subset P^C_{in}$ from (\ref{for-cla-main1-2}).  
Consequently, the set $S_H\cup (P^C_{in}\setminus P_H)\subset P^C_{in}$. Therefore, for any $p\in S_H\cup (P^C_{in}\setminus P_H)$ and any $q\in P^C_{out}$, $x_p\leq x_q$. Moreover, the set $S\setminus \Big(S_H\cup (P^C_{in}\setminus P_H)\Big)$
\begin{eqnarray}
&=&\Big(S_H\cup (P\setminus P_H)\Big)\setminus \Big(S_H\cup (P^C_{in}\setminus P_H)\Big)\nonumber\\
&=&(P\setminus P_H)\setminus (P^C_{in}\setminus P_H)\nonumber\\
&\underbrace{=}_{\text{by (\ref{for-cla-main1-2})}}& P\setminus P^C_{in}=P^C_{out}. \label{for-cla-main1-3}
\end{eqnarray}
Note $|P^C_{out}|=z$. As a consequence, $S^C_{in}$ should be exactly the set $S_H\cup (P^C_{in}\setminus P_H)$, and $S^C_{out}=P^C_{out}$.
\end{proof}

\begin{lemma}
\label{lem-main1}
If $P_H\setminus P^C_{in}=\emptyset$, (\ref{for-main}) is true.
\end{lemma}
\begin{proof}
Because the set $S^C_{in}$ is equal to $S_H\cup (P^C_{in}\setminus P_H)$ from Lemma~\ref{cla-main1}, the objective value $\mathcal{K}^{-z}_1(S, C)=\frac{1}{n-z}\Big(\sum_{p\in S_H}w(p) x_p+\sum_{p\in P^C_{in}\setminus P_H} x_p\Big)=$
\begin{eqnarray}
\frac{1}{n-z}\Big(\sum^N_{i=0}\frac{|P\cap H_i|}{|S_i|}\sum_{p\in S_i}x_p+\sum_{p\in P^C_{in}\setminus P_H} x_p\Big). \label{for-lem-main1-1}
\end{eqnarray}
From Lemma~\ref{the-sample}, the value of (\ref{for-lem-main1-1}) is no larger than
\begin{eqnarray}
&\leq&\frac{1}{n-z}\Big(\sum_{p\in P_H}x_p+O(\epsilon)(n-z)(\mathcal{K}^{-z}_1(P, \tilde{C})+L)\nonumber\\
&&+\sum_{p\in P^C_{in}\setminus P_H} x_p\Big).\label{for-lem-main1-2}
\end{eqnarray}
Note that $P_H\setminus P^C_{in}=\emptyset$, and thus the sum of the two terms $\sum_{p\in P_H}x_p$ and $\sum_{p\in P^C_{in}\setminus P_H} x_p$ in (\ref{for-lem-main1-2}) is $\sum_{p\in P^C_{in}}x_p$. Therefore, $\mathcal{K}^{-z}_1(S, C)\leq$
\begin{eqnarray}
&&\frac{1}{n-z}\sum_{p\in P^C_{in}}x_p + O(\epsilon) (\mathcal{K}^{-z}_1(P, \tilde{C})+L)\nonumber\\
&=&\mathcal{K}^{-z}_1(P, C) +O(\epsilon) (\mathcal{K}^{-z}_1(P, \tilde{C})+L).
\end{eqnarray}
Similarly, we have $\mathcal{K}^{-z}_1(S, C)\geq \mathcal{K}^{-z}_1(P, C)-O(\epsilon) (\mathcal{K}^{-z}_1(P, \tilde{C})+L)$. Thus, (\ref{for-main}) is true.
 \end{proof}

\subsection{Case (\rmnum{2}): $P_H\setminus P^C_{in}\neq\emptyset$}
\label{sec-case2}
Since $S\setminus P_H=P\setminus P_H$ are the outermost $(1+1/\epsilon)z$ points to $\tilde{C}$, we have the following claim first (due to the space limit, please refer to our supplement for the detailed proof).
\begin{claim}
\label{cla-main2-2}
Either $ S^C_{in}\setminus P_H \subseteq  P^C_{in}\setminus P_H$ or $P^C_{in}\setminus P_H\subseteq S^C_{in}\setminus P_H$ is true.
\end{claim}

\begin{lemma}
\label{lem-main2-1}
If $P_H\setminus P^C_{in}\neq\emptyset$, we have $x_p\leq 2^N r+L$ for any $p\in S^C_{in}\cup P^C_{in}\cup P_H$.
\end{lemma}
\begin{proof}
We consider the points in the three parts $P_H$, $P^C_{in}$, and $S^C_{in}$ separately. 

\textbf{(1)} Due to~(\ref{for-vbound}), we have $x_p\leq 2^N r+L$ for any $p\in  P_H$.  

\textbf{(2)} Arbitrarily select one point $p_0$ from $P_H\setminus P^C_{in}$. By (\ref{for-vbound}) again, we have $x_{p_0}\leq 2^N r+L$. Also, because $P_H\setminus P^C_{in}\subset P^C_{out}$, we directly have $x_p\leq x_{p_0}$ for any $p\in P^C_{in}$. Namely, $x_{p}\leq  2^N r+L$ for any $p\in P^C_{in}$. 

\textbf{(3)} Below, we consider the points in $S^C_{in}$. 
If $w(S^C_{in}\cap P_H)> \big|P^C_{in}\cap P_H\big|$, {\em i.e.}, $P_H$ contains more inliers of $S$ than that of $P$, then the outer region $H_{out}$ should contain less inliers of $S$ than that of $P$.   Thus, from Claim~\ref{cla-main2-2}, we have $S^C_{in}\setminus P_H \subseteq  P^C_{in}\setminus P_H$. 
Hence, $S^C_{in}=(S^C_{in}\setminus P_H)\cup (S^C_{in}\cap P_H) \subseteq (P^C_{in}\setminus P_H)\cup P_H=P^C_{in} \cup P_H$. 
From \textbf{(1)} and \textbf{(2)}, we know $x_{p}\leq 2^N r+L$ for any $p\in S^C_{in}$.

Else, $w(S^C_{in}\cap P_H)\leq \big|P^C_{in}\cap P_H\big|$. Then $w(S^C_{in}\cap S_H)\leq \big|P^C_{in}\cap P_H\big|$ since $S^C_{in}\cap P_H=S^C_{in}\cap S_H$. Because $w(S_H)=|P_H|$, we have
\begin{eqnarray}
w(S_H\setminus S^C_{in})\geq |P_H\setminus P^C_{in}|. 
\end{eqnarray}
Also, the assumption $P_H\setminus P^C_{in}\neq \emptyset$ implies $w(S_H\setminus S^C_{in})\geq |P_H\setminus P^C_{in}|>0$, {\em i.e.,} 
\begin{eqnarray}
S_H\setminus S^C_{in}\neq \emptyset.\label{for-lem-main2-1-3}
\end{eqnarray}
Arbitrarily select one point $p_0$ from $S_H\setminus S^C_{in}$. We know $x_{p_0}\leq  2^N r+L$ since $p_0\in S_H\setminus S^C_{in}\subset P_H$. 
Also, for any point $p\in S^C_{in}$, we have $x_p\leq x_{p_0}$ because $p_0\in S_H\setminus S^C_{in}\subset S^C_{out}$. Therefore $x_p\leq  2^N r+L$. 
 \end{proof}

\begin{lemma}
\label{lem-main2}
If $P_H\setminus P^C_{in}\neq\emptyset$, (\ref{for-main}) is true.
\end{lemma}
\begin{proof}
We prove the upper  bound of $\mathcal{K}^{-z}_1(S, C)$ first. We analyze the clustering costs of the two parts $S^C_{in}\cap S_H$ and $S^C_{in}\setminus S_H$ separately. 
\begin{eqnarray}
\mathcal{K}^{-z}_1(S, C)=\frac{1}{n-z}\Big(\underbrace{\sum_{p\in S^C_{in}\cap S_H}w(p) x_p}_{\textbf{(a)}}+\underbrace{\sum_{p\in S^C_{in}\setminus S_H}x_p}_{\textbf{(b)}}\Big). \nonumber
\end{eqnarray}
Note the points of $S^C_{in}\setminus S_H$ have unit-weight (since $S^C_{in}\setminus S_H\subseteq P\setminus P_H$ are the points from the outermost $(1+\frac{1}{\epsilon})z$ points of $P$). Obviously, the part \textbf{(a)} is no larger than
\begin{eqnarray}
&&\sum_{p\in  S_H}w(p) x_p=\sum^N_{i=0}\frac{|P\cap H_i|}{|S_i|}\sum_{p\in S_i}x_p\nonumber\\
&\leq&\sum_{p\in P_H}x_p+  O(\epsilon) (n-z)(\mathcal{K}^{-z}_1(P, \tilde{C})+L)\label{for-lem-main2-2-2}
\end{eqnarray}
from Lemma~\ref{the-sample}. The set $P_H$ consists of two parts $P_H\cap P^C_{in}$ and $P_H\setminus P^C_{in}$. From Lemma~\ref{lem-main2-1} and the fact $|P_H\setminus P^C_{in}|\leq |P^C_{out}|= z$, we know $\sum_{p\in P_H\setminus P^C_{in}}x_p\leq z(2^N r+L)$. Thus, the upper bound of the part \textbf{(a)} becomes 
\begin{eqnarray}
&\sum_{p\in P_H\cap P^C_{in}} x_p+  z(2^N r+L)+\nonumber\\
&O(\epsilon) (n-z)(\mathcal{K}^{-z}_1(P, \tilde{C})+L).\label{for-lem-main2-2-3}
\end{eqnarray}

To bound the part \textbf{(b)}, we consider the size $|S^C_{in}\setminus S_H|$. 
Since the total weight of outliers is $z$, $w\big(S_H\cap S^C_{in}\big)$
\begin{eqnarray}
&=&w(S_H)-w\big(S_H\cap S^C_{out}\big)\geq w(S_H)-z\nonumber\\
&=&\big|P_H \big|-z\geq\big|P_H\cap P^C_{in}\big|-z. \label{for-lem-main2-2-4}
\end{eqnarray}
Together with the fact $w\big(S_H\cap S^C_{in}\big)+\big|S^C_{in}\setminus S_H\big|=\big|P_H\cap P^C_{in}\big|+\big|P^C_{in}\setminus P_H\big|=n-z$, we have
\begin{eqnarray}
\big|S^C_{in}\setminus S_H\big|\leq \big|P^C_{in}\setminus P_H\big|+z.
\end{eqnarray}
Therefore $\Big|\big(S^C_{in}\setminus S_H\big)\setminus\big(P^C_{in}\setminus P_H\big)\Big|\leq z$ from Claim~\ref{cla-main2-2}. Through Lemma~\ref{lem-main2-1} again, we know that the part \textbf{(b)} is no larger than $\sum_{p\in P^C_{in}\setminus P_H}x_p+\Big|\big(S^C_{in}\setminus S_H\big)\setminus\big(P^C_{in}\setminus P_H\big)\Big|\cdot(2^N r+L)$
\begin{eqnarray}
\leq\sum_{p\in P^C_{in}\setminus P_H}x_p+z(2^N r+L).\label{for-lem-main2-2-5}
\end{eqnarray}

Putting (\ref{for-lem-main2-2-3}) and (\ref{for-lem-main2-2-5}) together, we have $\mathcal{K}^{-z}_1(S, C)\leq$
\begin{eqnarray}
 &&\mathcal{K}^{-z}_1(P, C)+ O(\epsilon)(\mathcal{K}^{-z}_1(P, \tilde{C})+L)\nonumber\\
&& +\frac{2z}{n-z}(2^N r+L).\label{for-upper}
\end{eqnarray}

Recall $\gamma=z/(n-z)$ in Algorithm~\ref{alg-core}. If $\gamma\geq\epsilon$, the size of our coreset $S$ is at least $(1+1/\epsilon)z\geq n$; that is, $S$ contains all the points of $P$. For the other case $\epsilon>\gamma$, together with (\ref{for-r}), the term $\frac{2z}{n-z}(2^N r+L)$ in (\ref{for-upper}) is at most
\begin{eqnarray}
O(\epsilon)(\mathcal{K}^{-z}_1(P, \tilde{C})+L).
\end{eqnarray}
Overall, $\mathcal{K}^{-z}_1(S, C)\leq  \mathcal{K}^{-z}_1(P, C)+ O(\epsilon)(\mathcal{K}^{-z}_1(P, \tilde{C})+L)$ via (\ref{for-upper}).
So we complete the proof for the upper bound.

Now we consider the lower bound of $\mathcal{K}^{-z}_1(S, C)$. Denote by $X=S_H\cup (P^C_{in}\setminus P_H)$ and $Y=X\setminus S^C_{in}$. Obviously, 
\begin{eqnarray}
\mathcal{K}^{-z}_1(S, C)&\geq&\frac{1}{n-z}\Big(\underbrace{\sum_{p\in X}w(p) x_p}_{\textbf{(c)}} -\underbrace{\sum_{p\in Y}w(p) x_p}_{\textbf{(d)}}\Big).\nonumber
\end{eqnarray}
From Lemma~\ref{the-sample}, the part \textbf{(c)} is at least 
\begin{eqnarray}
\sum_{p\in P_H}x_p-O(\epsilon)(\mathcal{K}^{-z}_1(P, \tilde{C})+L)+\sum_{p\in P^C_{in}\setminus P_H}x_p\nonumber\\
\geq \sum_{p\in P^C_{in}}x_p-O(\epsilon)(\mathcal{K}^{-z}_1(P, \tilde{C})+L). \label{for-lem-main2-2-6}
\end{eqnarray}
Further, since $w(Y)\leq z$ and $Y\subseteq X\subseteq  P^C_{in}\cup P_H$, the part \textbf{(d)} is no larger than 
$z(2^N r+L)$  
from  Lemma~\ref{lem-main2-1}.
Using the similar manner for proving the upper bound, we know that $\mathcal{K}^{-z}_1(S, C)\geq  \mathcal{K}^{-z}_1(P, C)- O(\epsilon)(\mathcal{K}^{-z}_1(P, \tilde{C})+L)$. 
 \end{proof}

\section{Linear Regression with Outliers}
\label{sec-lr}
In this section, we consider the problem of linear regression with outliers. Our algorithm and analysis are for the objective function $\mathcal{LR}^{-z}_1$, and the ideas can be extended to handle the objective function $\mathcal{LR}^{-z}_2$.

\begin{algorithm}[tb]
   \caption{\sc{Layered Sampling for Lin1-Outlier}}
   \label{alg-core-lr}
\begin{algorithmic}
  \STATE {\bfseries Input:} An instance $P\subset\mathbb{R}^d$ of linear regression with $z$ outliers, a solution $\tilde{h}=(\tilde{h}_1, \cdots, \tilde{h}_k)$, and two parameters $\epsilon, \eta\in (0,1)$.
   \STATE
\begin{enumerate}
\item Let $\gamma=z/(n-z)$ and $N=\lceil\log \frac{1}{\gamma}\rceil$. Compute the value $r$ satisfying (\ref{for-frame3}). 
\item As described in (\ref{for-frame8}), (\ref{for-frame9}), and (\ref{for-frame10}), the space is partitioned into $N+2$ layers $H_0, H_1, \cdots, H_N$ and $H_{out}$.
\item Randomly sample $\min\Big\{O(\frac{1}{\epsilon^2}d\log\frac{d}{\epsilon}\log\frac{N}{\eta}), |P\cap H_i|\Big\}$ points, denoted by $S_{i}$, from $P\cap H_i$ for $0\leq i\leq N$.
\item For each point $p\in S_{i}$, set its weight to be $|P\cap H_i|/|S_i|$; let $S_H=\cup^N_{i=0}S_i$.
\end{enumerate}
 
  \STATE {\bfseries Output} $S=S_H\cup (P\cap H_{out})$.
\end{algorithmic}
\end{algorithm}

\begin{theorem}
\label{the-main-lr}
Algorithm~\ref{alg-core-lr} returns a point set $S$ having the size $|S|=\tilde{O}(\frac{1}{\epsilon^2}d)+(1+\frac{1}{\epsilon})z$. Moreover, with probability at least $1-\eta$, for any $L>0$ and any solution $h\in \tilde{h}\pm L$, we have 
\begin{eqnarray}
\mathcal{LR}_{1}^{-z} (S, h) \in\mathcal{LR}_{1}^{-z} (P, h)\pm\epsilon\big(\mathcal{LR}_{1}^{-z} (P, \tilde{h})+L\big).\label{for-main-lr}
\end{eqnarray}
Here, $S$ is a weighted instance of linear regression with outliers, and the total weight of outliers is $z$ (see Remark~\ref{re-weight}).
\end{theorem}
We still use $P_H$ to denote the set $\cup^N_{i=0}(P\cap H_i)$. 
First, we need to prove that $S_H$ is a good approximation of $P_H$. Given a hyperplane $h$, we define a random variable $x_p=|Res(p, h)|$ for each $p\in P$. If $p\in H_i$ for $0\leq i\leq N$, similar to (\ref{for-vbound}) and (\ref{for-vboundlower}), we have the following bounds for $x_p$: 
$x_p\leq 2^i r+L$; 
      $ x_p\geq\max\{ 2^{i-1}r-L, 0\}$   if $i\geq 1$ and 
        $ x_p\geq 0$   if $i=0$.

Then, we can apply the similar idea of Lemma~\ref{the-sample} to obtain the following lemma, where the only difference is about the discretization on $\tilde{h}\pm L$. Recall that $\tilde{h}$ is defined by the coefficients $\tilde{h}_1, \cdots, \tilde{h}_d$ and the input set $P$ is normalized within the region $\mathcal{R}_D$. We build a grid inside each vertical segment $\overline{l_j u_j}$ for $0\leq j\leq d-1$, where $l_0=(0, \cdots, 0, \tilde{h}_d-L)$, $u_0=(0, \cdots, 0, \tilde{h}_d+L)$, and 
\begin{eqnarray}
 l_j=(0, \cdots, 0,\underbrace{D}_{j-th}, 0, \cdots, 0,  \tilde{h}_j D+\tilde{h}_d-L), \\
u_j=(0, \cdots, 0, \underbrace{D}_{j-th}, 0, \cdots, 0,   \tilde{h}_j D+\tilde{h}_d+L)
 \end{eqnarray}
for $j\neq 0$; the grid length is $\frac{\epsilon}{2d} L$. Denote by $G_j$ the set of grid points inside the segment $\overline{l_j u_j}$. Obviously, $\mathcal{G}=G_0\times G_1\times\cdots\times G_{d-1}$ contains $(\frac{4d}{\epsilon})^{d}$ $d$-tuple points, and each tuple determines a $(d-1)$-dimensional hyperplane in $\tilde{h}\pm L$; moreover, we have the following claim (see the proof in our supplement).

\begin{claim}
\label{cla-lr}
For each $h\in \tilde{h}\pm L$, there exist a hyperplane $h'$ determined by a $d$-tuple points from $\mathcal{G}$, such that $|Res(p, h)-Res(p, h')|\leq \epsilon L$ for any $p\in P$. 
\end{claim}

\begin{lemma}
\label{the-sample-lr}
Let $S_0, S_1, \cdots, S_N$ be the samples obtained in Algorithm~\ref{alg-core-lr}. Then, with probability $1-\eta$,
\begin{eqnarray}
&&\frac{1}{n-z}\big|\sum^N_{i=0}\frac{|P\cap H_i|}{|S_i|}\sum_{p\in S_i}x_p-\sum_{p\in P_H}x_p\big|\nonumber\\
&\leq& O(\epsilon)\Big(\mathcal{LR}^{-z}_1(P, \tilde{h})+L\Big) \label{for-the-sample-lr}
\end{eqnarray}
for any $h\in\tilde{h}\pm L$.
\end{lemma}

We fix a solution $h\in \tilde{h}\pm L$. Similar to the proof of Theorem~\ref{the-main} in Section~\ref{sec-construction}, we also consider the two parts $P^h_{in}$ and $P^h_{out}$ of $P$ partitioned by $h$, where $P^h_{out}$ is the $z$ farthest points to the hyperplane $h$ ({\em i.e.,} the outliers) and $P^h_{in}=P\setminus P^h_{out}$. Similarly, $S$ is also partitioned into two parts $S^h_{in}$ and $S^h_{out}$ by $h$, where $S^h_{out}$ is the set of outliers with total weights $z$. For case (\rmnum{1}) $P_H\setminus P^h_{in}=\emptyset$ and case (\rmnum{2}) $P_H\setminus P^h_{in}\neq\emptyset$, we can apply almost the identical ideas in Section~\ref{sec-case1} and \ref{sec-case2} respectively to prove (\ref{for-main-lr}).

\section{Experiments}

For both Algorithm 1 and 2, we need to compute an initial solution $\tilde{C}$ or $\tilde{h}$ first. 
For $ k $-median/means clustering, we run the algorithm of Local Search with Outliers from \cite{Gupta:2017:LSM:3067421.3067425} on a small sample of size $O(k)$ to obtain the $ k $ initial centers. We do not directly use the $k$-means~+~+ method \cite{Arthur:2007:KAC:1283383.1283494} to seed the $k$ initial centers because it is sensitive to outliers. 
For linear regression, we  run the standard linear regression algorithm on a small random sample of size $O(z)$ to compute an initial solution.

\paragraph{Different coreset methods}
Each of the following methods returns a weighted set as the coreset, and then we run the alternating minimization algorithm $k$-means{-}{-} \cite{Chawla2013kmeansAU} (or  \cite{ShenS19} for linear regression) on it to obtain the solution. For fairness, we keep the coresets from different methods to have the same coreset size for each instance.
\begin{itemize}
	\item Layered Sampling (\texttt{LaySam}).
	{\em i.e.,} Algorithm 1 and 2 
	proposed in this paper. 
	
	\item  Uniform Sampling (\texttt{UniSam}). The most natural and simple method is to take a sample $S$ uniformly at random from the input data set $P$, where each sampled point has the weight $|P|/|S|$. 
	\item Uniform Sampling + Nearest Neighbor Weight (\texttt{NN} ) \cite{Gupta:2017:LSM:3067421.3067425,DBLP:conf/nips/ChenA018}. Similar to \texttt{UniSam}, we also take a random sample $S$ from the input data set $P$. For each $p\in P$, we assign it to its nearest neighbor in $S$; for each $q\in S$, we set its weight to be the number of points assigned to it.  
	
	\item \texttt{Summary} \cite{DBLP:conf/nips/ChenA018}. It is a method to construct the coreset for $k$-median/means clustering with outliers by successively sampling and removing points from the original data until the number of the remaining points is small enough. 
	
\end{itemize}

The above \texttt{LaySam,UniSam} and \texttt{NN} are also used for linear regression in our experiments. 
We run 10 trials for each case and take the average. All the experimental results were obtained on a Ubuntu server with 2.4GHz E5-2680V4 and 256GB main memory; the algorithms were implemented in Matlab R2018b. 

\paragraph{Performance Measures}
The following measures  will be taken into account in the experiment.
\begin{itemize}
	\item \texttt{$ \ell_1 $-loss}:  $ \mathcal{K}_1^{-z} $ or $ \mathcal{LR}_1^{-z} $.
	\item \texttt{$ \ell_2 $-loss}: $ \mathcal{K}_2^{-z} $ or $ \mathcal{LR}_2^{-z} $.
	\item \texttt{recall/precision}. 
	Let $ O^* $ and $ O $ be the sets of outliers with respect to  the optimal solution and our obtained solution, respectively. \texttt{recall} $=\frac{|O\cap O^*|}{|O^*|} $ and \texttt{precision} $ =\frac{|O\cap O^*|}{|O|}  $. Since $|O^*|=|O|=z$, \texttt{recall} $ = $ \texttt{precision} $=\frac{|O\cap O^*|}{z} $.
	\item \texttt{pre-recall}. It indicates the proportion of $ O^* $ that are included in the coreset. Let $S$ be the coreset and \texttt{pre-recall} $ =\frac{|S\cap O^*|}{|O^*|} $. We pay attention  in particular to this measure, because the outliers could be quite important and may reveal some useful information ({\em e.g.,} anomaly detection). For example, as for clustering, if we do not have any prior knowledge of a given biological data, the outliers could be from an unknown tiny species. Consequently it is more preferable to keep such information when compressing the data. More detailed discussion on the significance of outliers can be found in \cite{beyer2007robust,zimek2012survey,DBLP:conf/swat/Moitra18,DBLP:journals/cacm/GoodfellowMP18}.
\end{itemize}

\paragraph{Datasets}
We consider the following datasets in our experiments. 
\begin{itemize}
	\item \texttt{syncluster} We generate the synthetic data as follows: Firstly we create $ k $ centers with each dimension randomly located in $ [0,100] $. 
	Then we generate the points following standard Gaussian distributions around  the centers. 
	\item \texttt{synregression} Firstly, we randomly set the $d$  coefficients of hyperplane $ h $ in $ [-5,+5] $ and construct $ \mathbf{x_i} $ in $ [0,10]^{d-1} $ by uniform sampling. Then let $ y_i$ be the inner product of $(\mathbf{x_i}, 1)$ and $ h $. Finally we randomly perturb each $ y_i $ by $ \mathcal{N}(0,1)$.
	\item \texttt{3DSpatial} ($n=434874$, $d=4$). This dataset was constructed by adding the elevation information to a 2D road network in North Jutland, Denmark~\cite{3d}.
	\item \texttt{covertype} ($n=581012$, $d=10$). It is a forest cover type dataset from Jock A. Blackard (UCI Machine Learning Repository), and we select its first 10 attributes. 
	\item \texttt{skin} ($n=245057$, $d=3$). The skin dataset is collected by randomly sampling B, G, R values from face images of various age groups and we select its first three dimension~\cite{skin}.
	\item  \texttt{SGEMM} ($n=241600$, $d=15$).  It contains the running times for multiplying two 2048 x 2048 matrices using a GPU OpenCL SGEMM kernel with varying parameters~\cite{SGEMM}.
	\item \texttt{PM2.5} ($n=41757$, $d=11$). It is a data set containing the PM2.5 data of US Embassy in Beijing~\cite{pm2.5}.
\end{itemize}
For each dataset, we randomly pick $ z $ points to be outliers by perturbing their locations in each dimension. 
We use a parameter $ \sigma $ to measure the extent of perturbation. For example, we consider the Gaussian distribution $ \mathcal{N}(0,\sigma) $ and uniform distribution $ [-\sigma,+\sigma] $. So the larger the parameter $ \sigma $ is, the more diffused the outliers will be. For simplicity, we use the notations in the form of \texttt{[dataset]-[distribution]-$ \sigma $} to indicate the datasets, {\em e.g.}, \texttt{syncluster-gauss-$\sigma$}.

\subsection{Coreset Construction Time}
We fix the coreset size and vary the data size $n$  of the synthetic datasets, and show the coreset construction times in Figure~\ref{expr-time} and \ref{expr-time2}. 
It is easy to see that the construction time of \texttt{NN}  is larger than other construction times by several orders of magnitude. 
 It is not out of expectation that \texttt{UniSam} is the fastest (because it does not need any operation except uniform random sampling). Our  \texttt{LaySam} lies in between \texttt{UniSam}  and \texttt{Summary}. 
\begin{figure}[H]
	\subfigure[]{\includegraphics[width=0.48\linewidth]{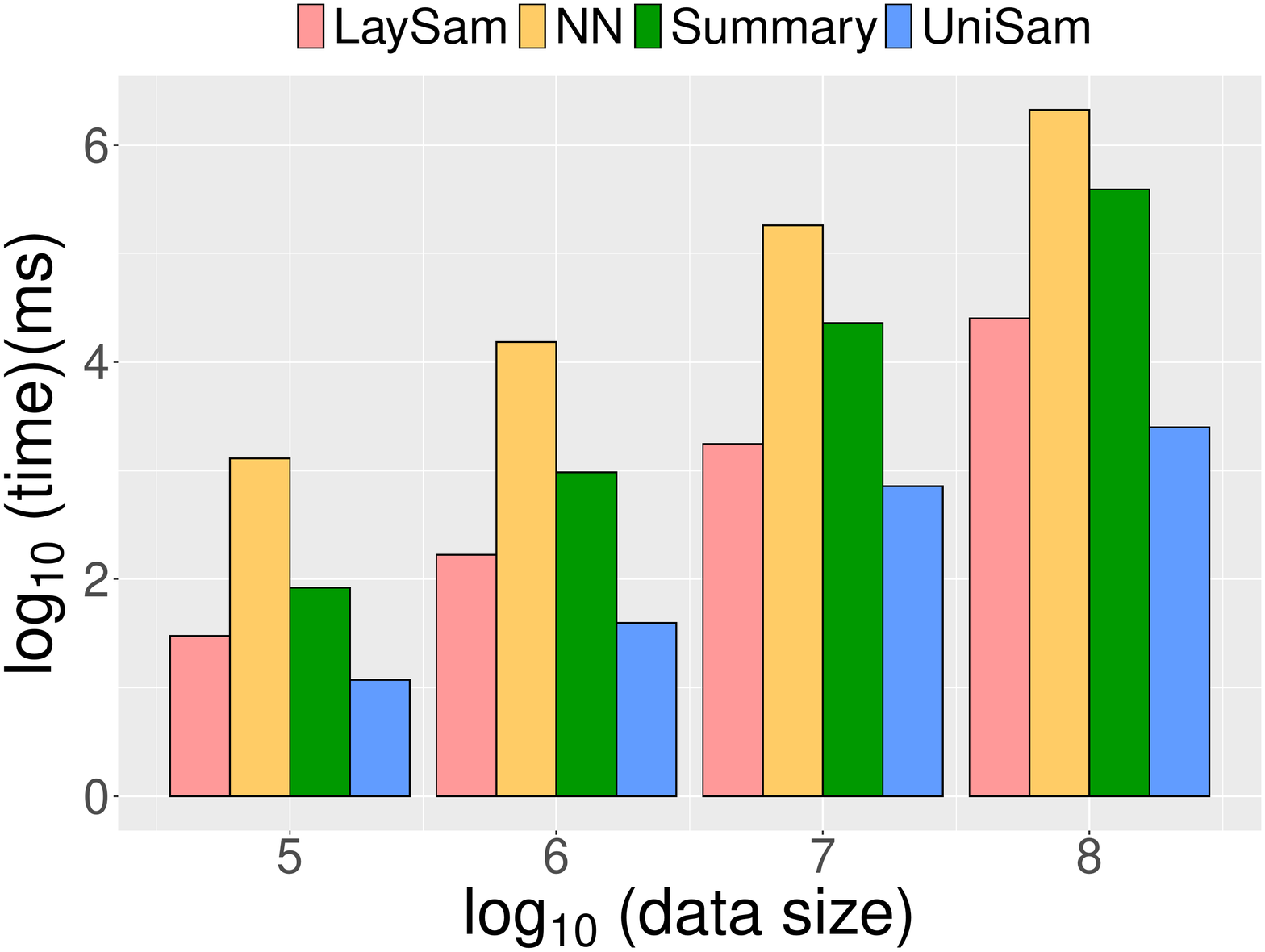}\label{expr-time}}
	\subfigure[]{\includegraphics[width=0.48\linewidth]{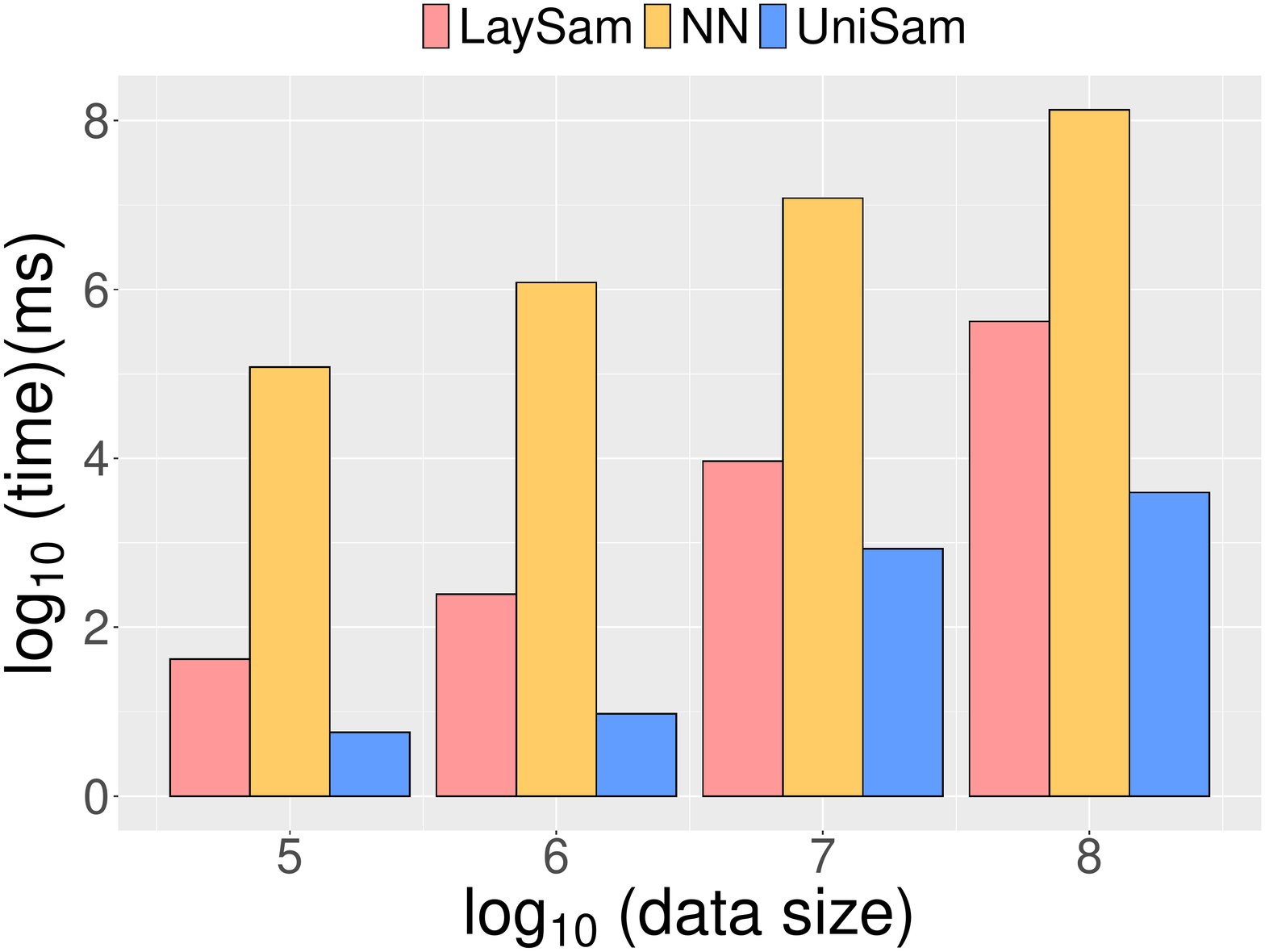}\label{expr-time2}}
		\caption{Coreset construction time {\em w.r.t.} data size. (a) Clustering: coreset size=$ 10^4 $, $ d=5 $ and $ k=10 $; (b) Linear Regression: coreset size=$ 10^4 $ and $ d=20 $.}
\end{figure}

\begin{figure}[H]
	\subfigure[]{\includegraphics[width=0.48\linewidth]{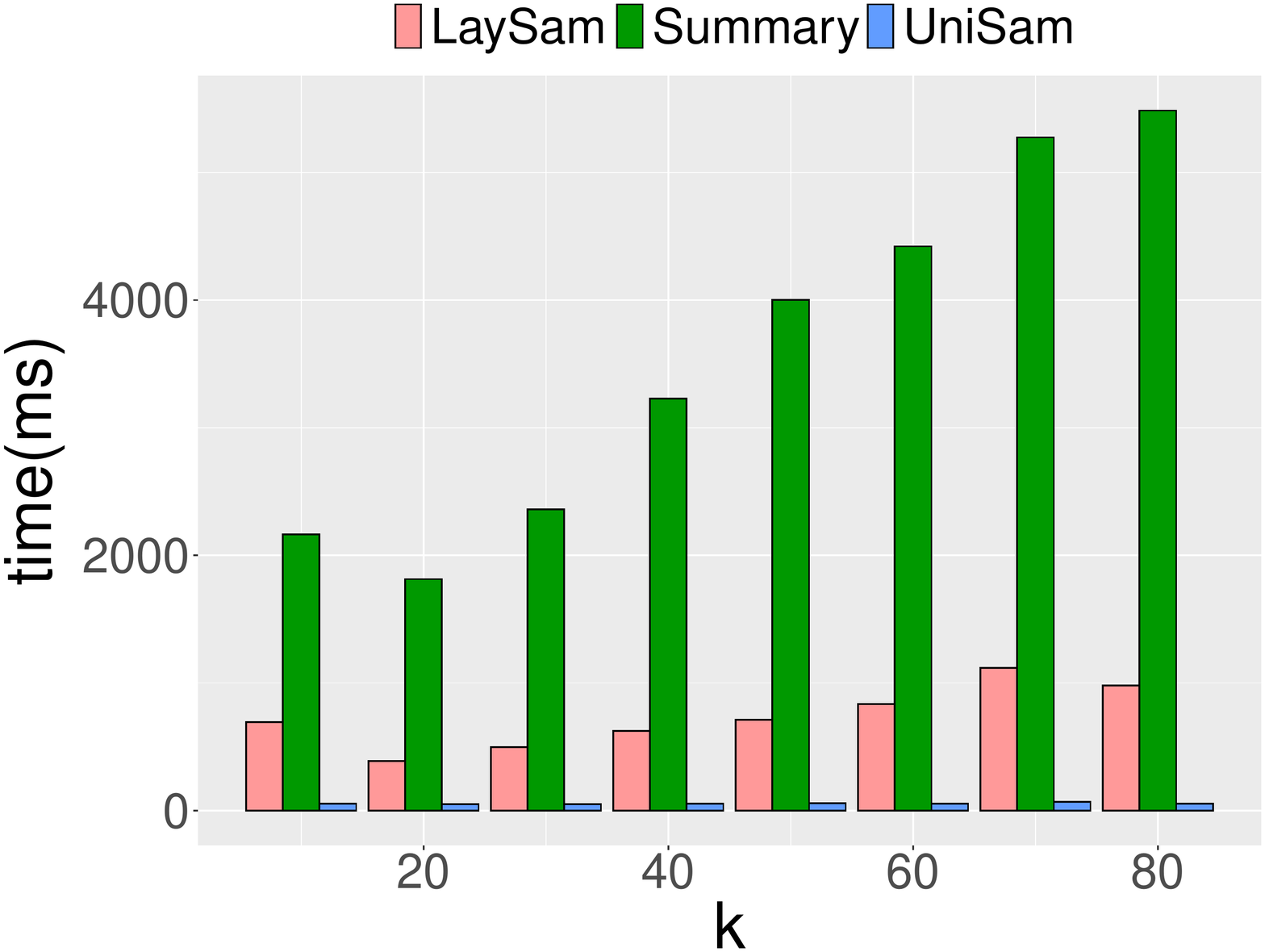}\label{expr-time-22}}	
		\subfigure[]{\includegraphics[width=0.48\linewidth]{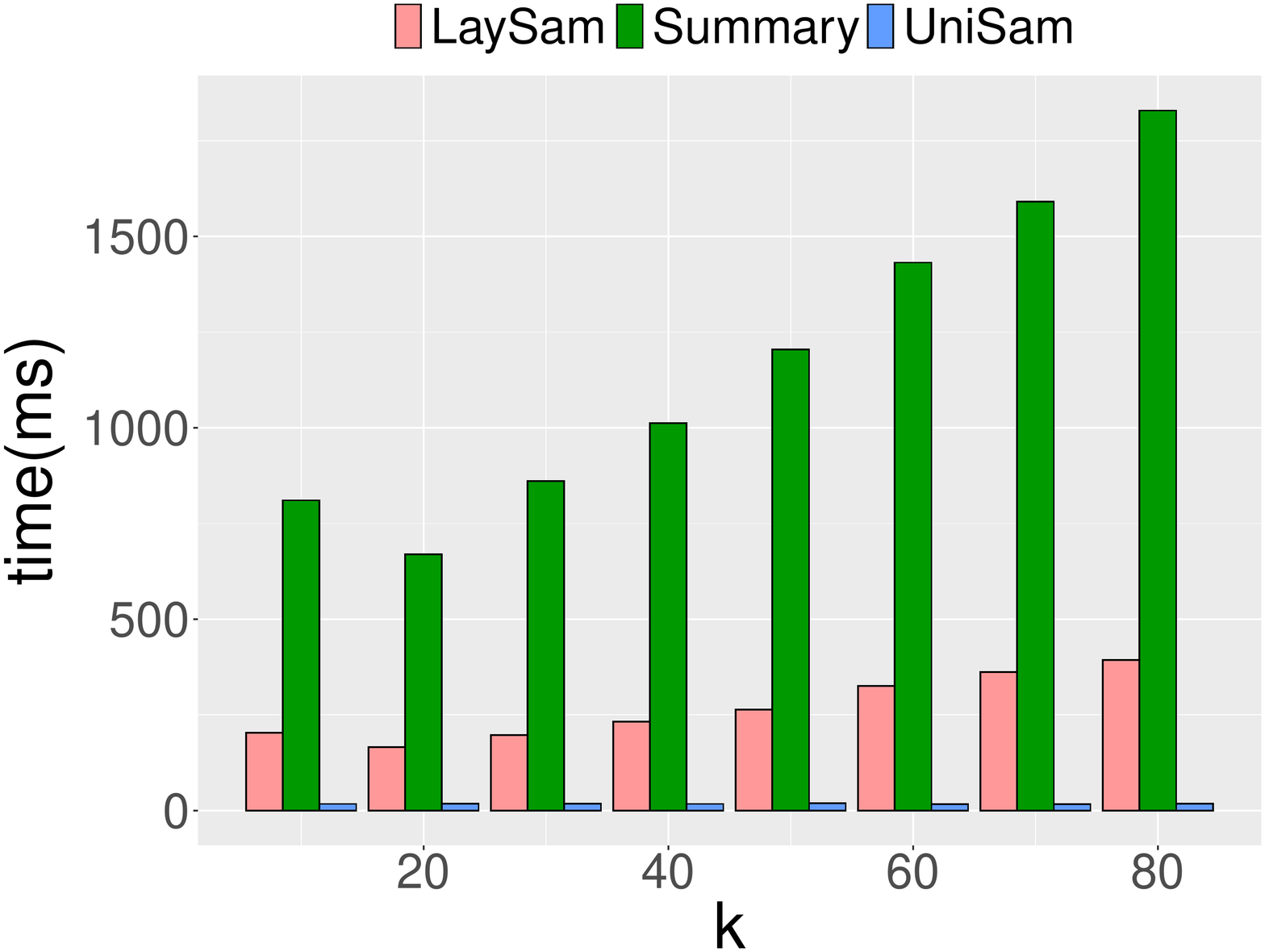}\label{expr-time-2}}
	\caption{Construction time {\em w.r.t.} $ k $. (a) \texttt{syncluster} with data size $n= 10^6 $, $ d=20 $; (b) \texttt{covertype}. }
\end{figure} 
We also study the influence of $k$ (for clustering) on the construction time by testing the synthetic datasets and the real-world dataset \texttt{covertype} (see Figure \ref{expr-time-22} and \ref{expr-time-2}).
\begin{figure*}
\begin{center}
	\subfigure[\texttt{syncluster-gauss-$ \sigma $}, with $ d=20 $, $ k=10$ ,$z=2\cdot10^4$ and coreset size$ =5\cdot 10^4 $.]{\includegraphics[width=0.3\linewidth]{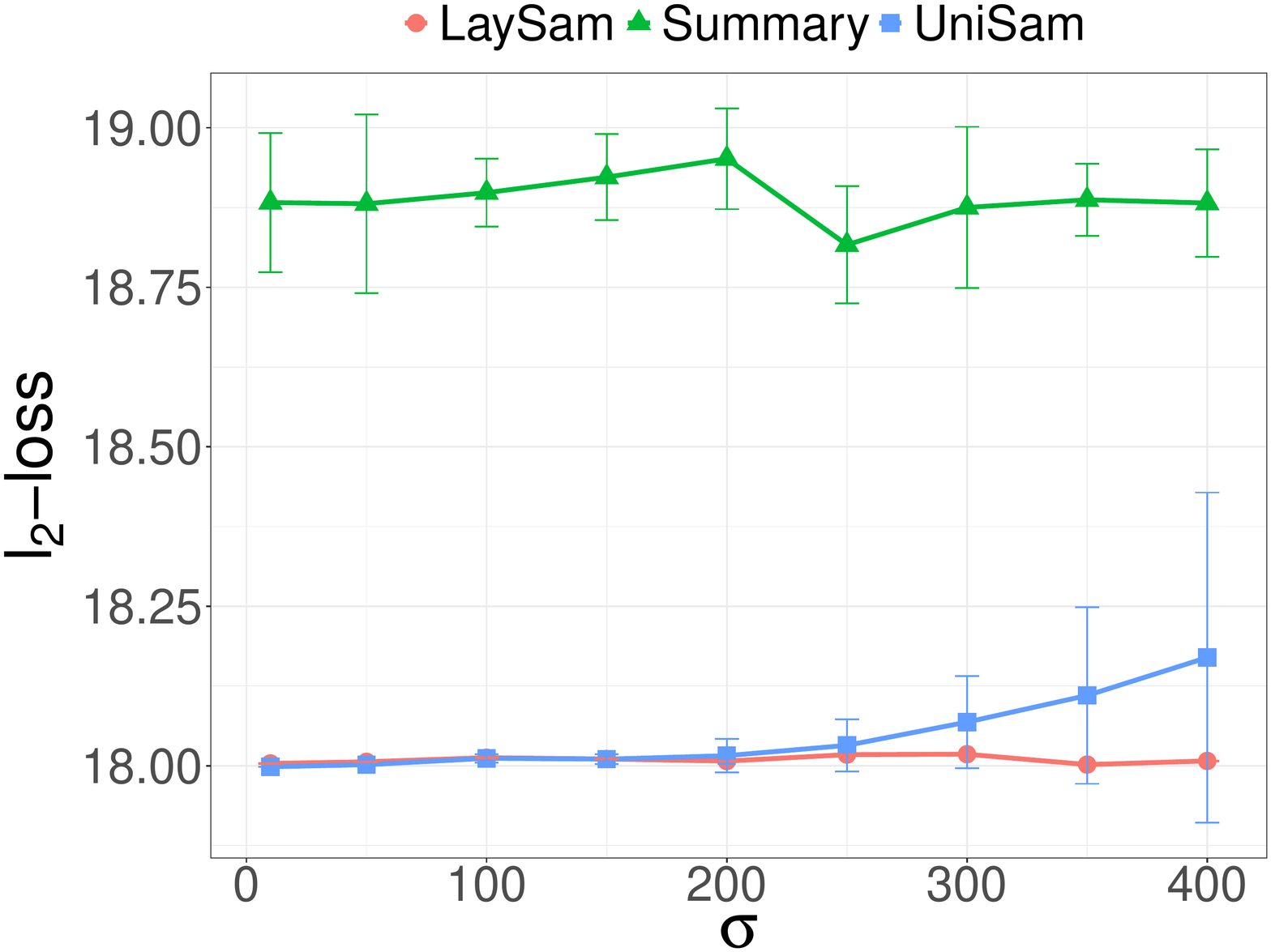}\label{error-cluster}}
		\hspace{0.18in}
	\subfigure[\texttt{3DSpatial-gauss-$ \sigma $}, with $ k=5$, $z=4000$ and coreset size$ =10^4 $.]{\includegraphics[width=0.3\linewidth]{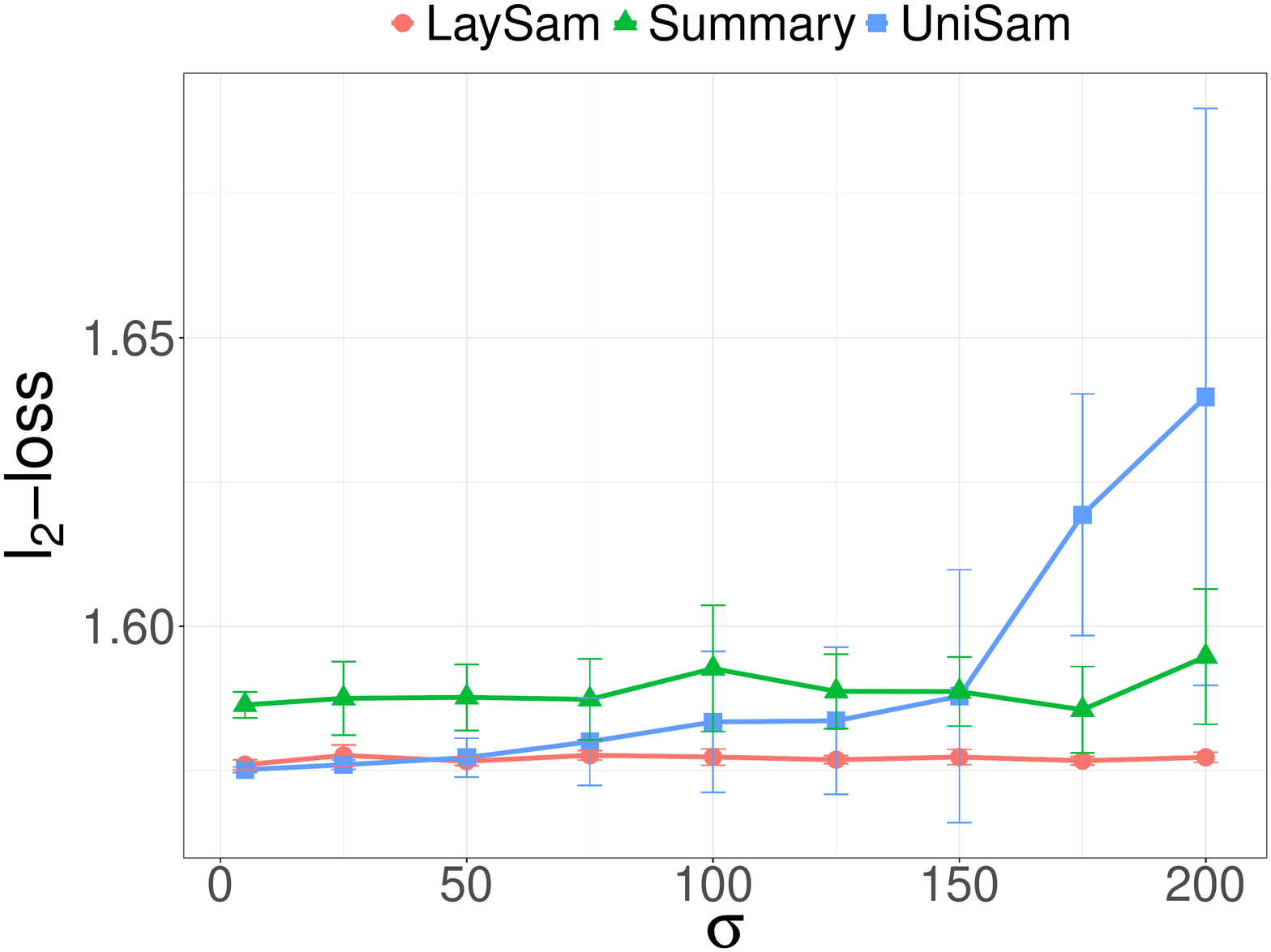}\label{error-cluster2}}
		\hspace{0.18in}
	\subfigure[\texttt{covertype-gauss-$ \sigma $}, with $ k=20 $, $z=10^4$ and coreset size$ =3\cdot10^4 $.]{\includegraphics[width=0.3\linewidth]{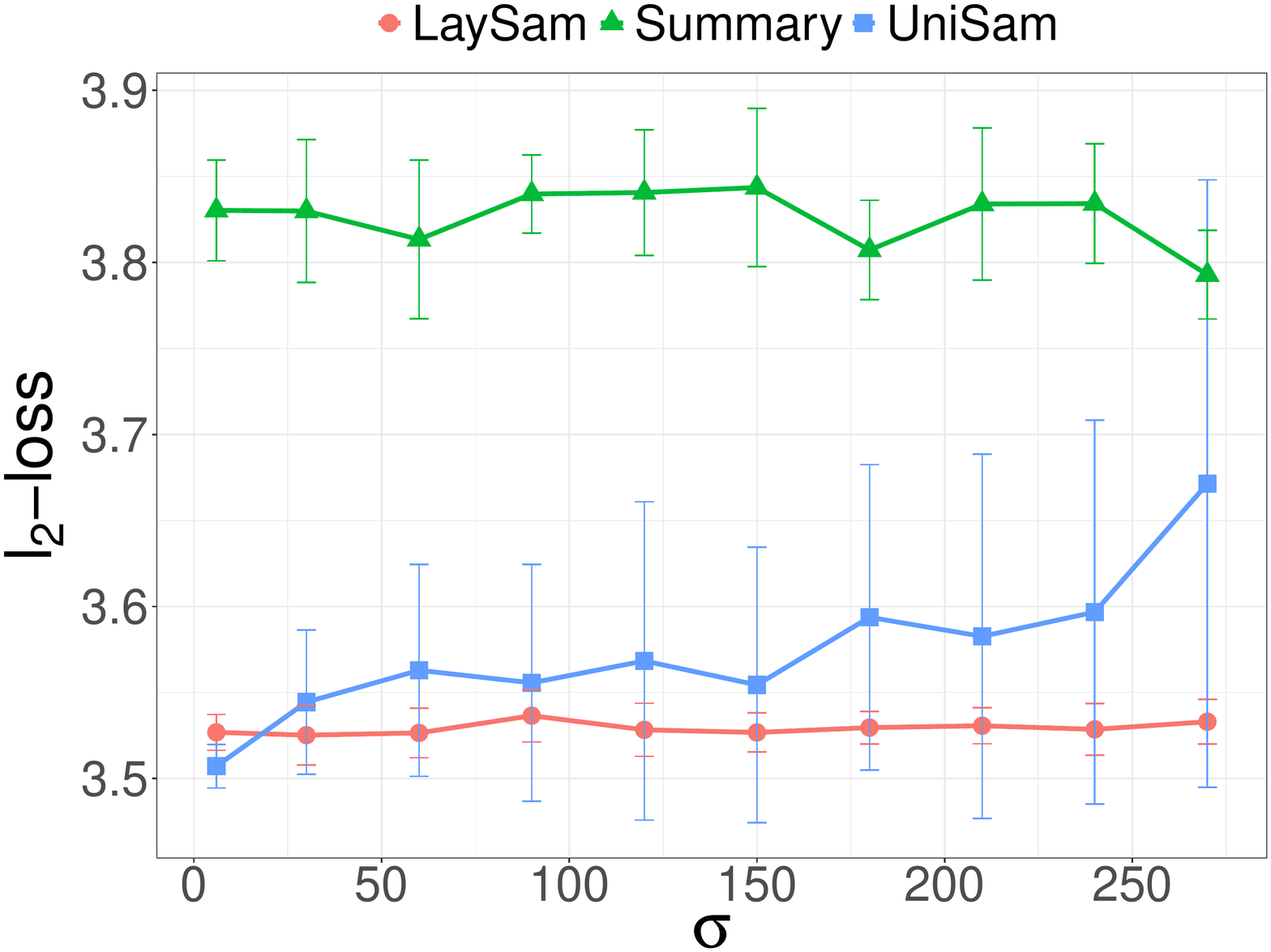}\label{error-cluster3}}
	\subfigure[\texttt{covertype-uniform-$ \sigma $}, with $ k=20 $, $z=10^4$ and coreset size$ =3\cdot10^4 $.]{\includegraphics[width=0.3\linewidth]{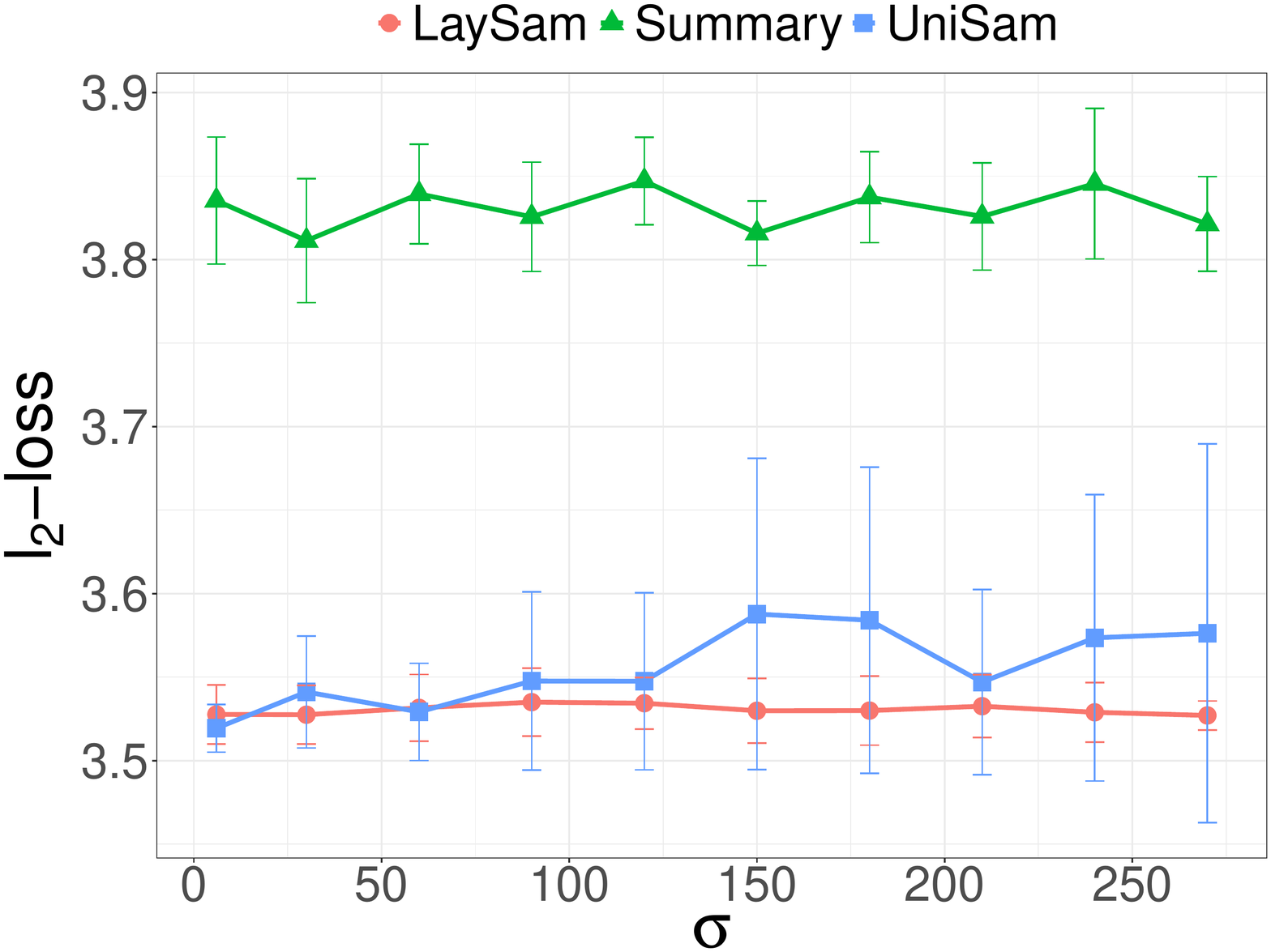}\label{error-cluster4}}
	\hspace{0.25in}
	\subfigure[\texttt{skin-uniform-$ \sigma $}, with $ k=40 $, $z=3000$ and coreset size$ =10^4 $.]{\includegraphics[width=0.3\linewidth]{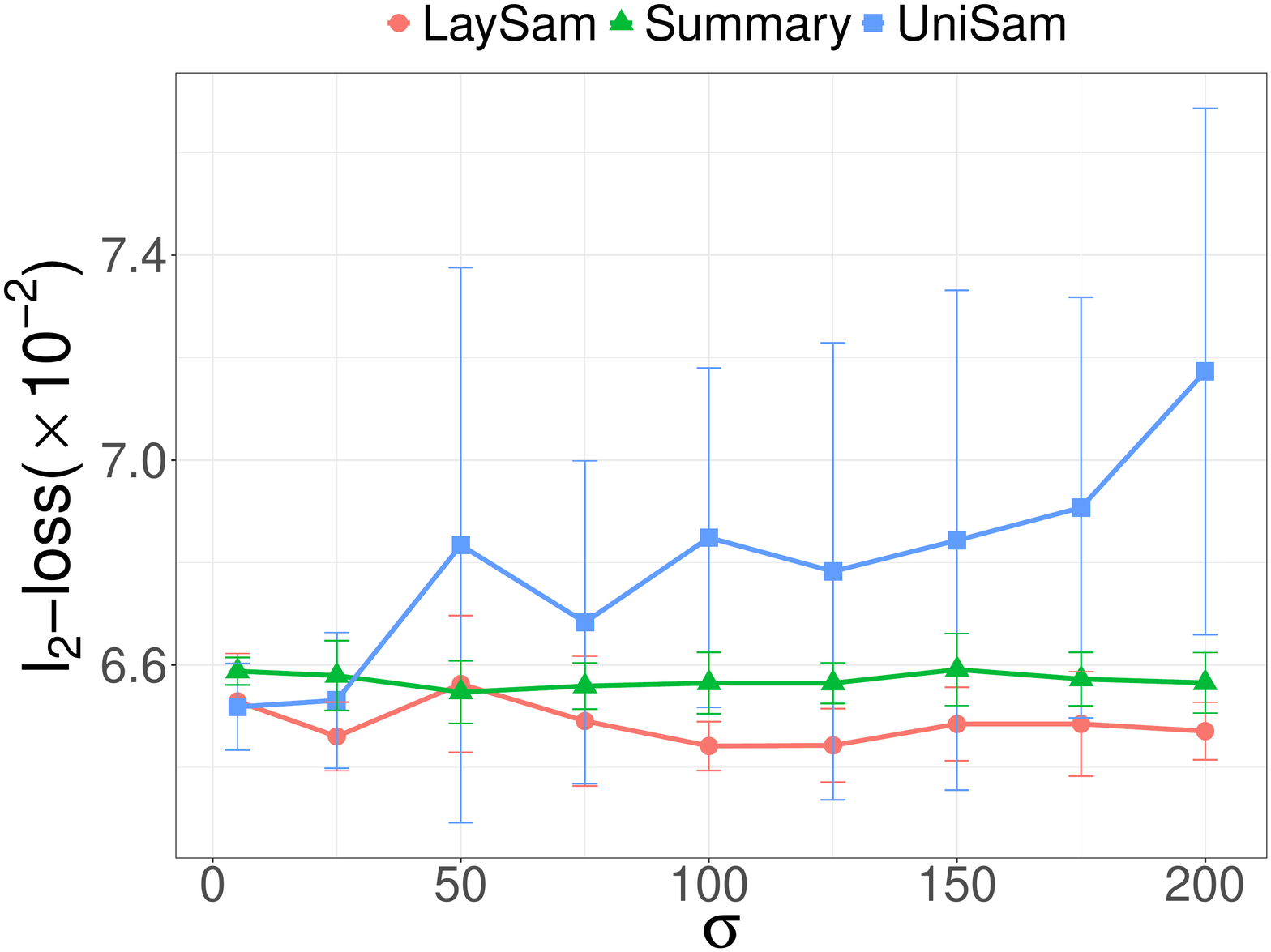}\label{error-cluster5}}
	\caption{Clustering: \texttt{$ \ell_2 $-loss} and stability {\em w.r.t.} $ \sigma $.}
	\end{center}
\end{figure*}

\begin{figure*}
\begin{center}
	\subfigure[\texttt{synregression-gauss-$ \sigma $}, with $ d=20 $, $ z=10^4 $ and coreset size=$ 2\cdot10^4 $.]{\includegraphics[width=0.3\linewidth]{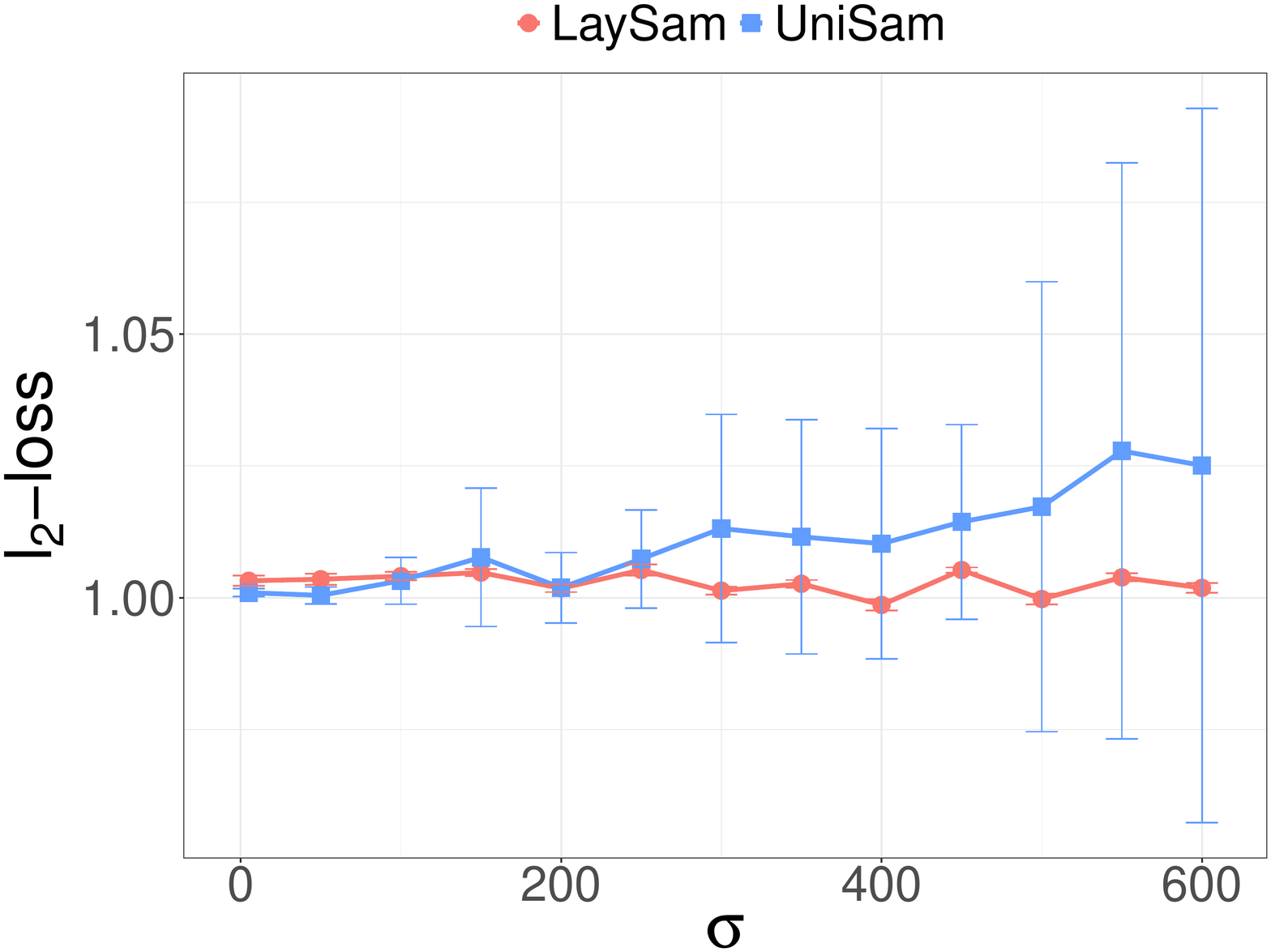}\label{error-lr}}
			\hspace{0.18in}
	\subfigure[\texttt{SGEMM-gauss-$ \sigma $}, with $ z=4000 $ and coreset size=$ 10^4 $.]{\includegraphics[width=0.3\linewidth]{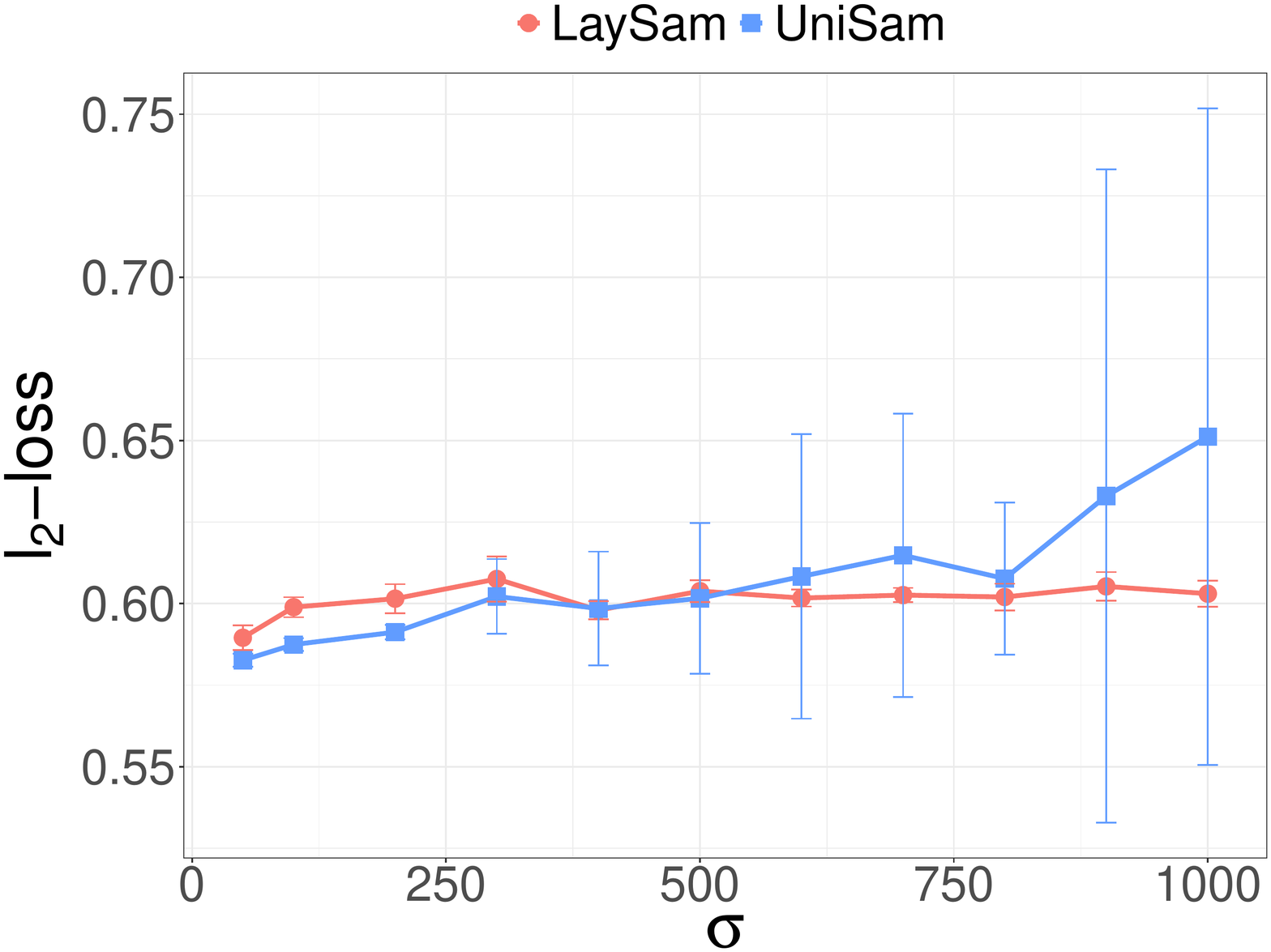}\label{error-lr2}}
			\hspace{0.18in}
	\subfigure[\texttt{PM2.5-uniform-$ \sigma $}, with $ z=1000 $ and coreset size=$ 3000 $.]{\includegraphics[width=0.3\linewidth]{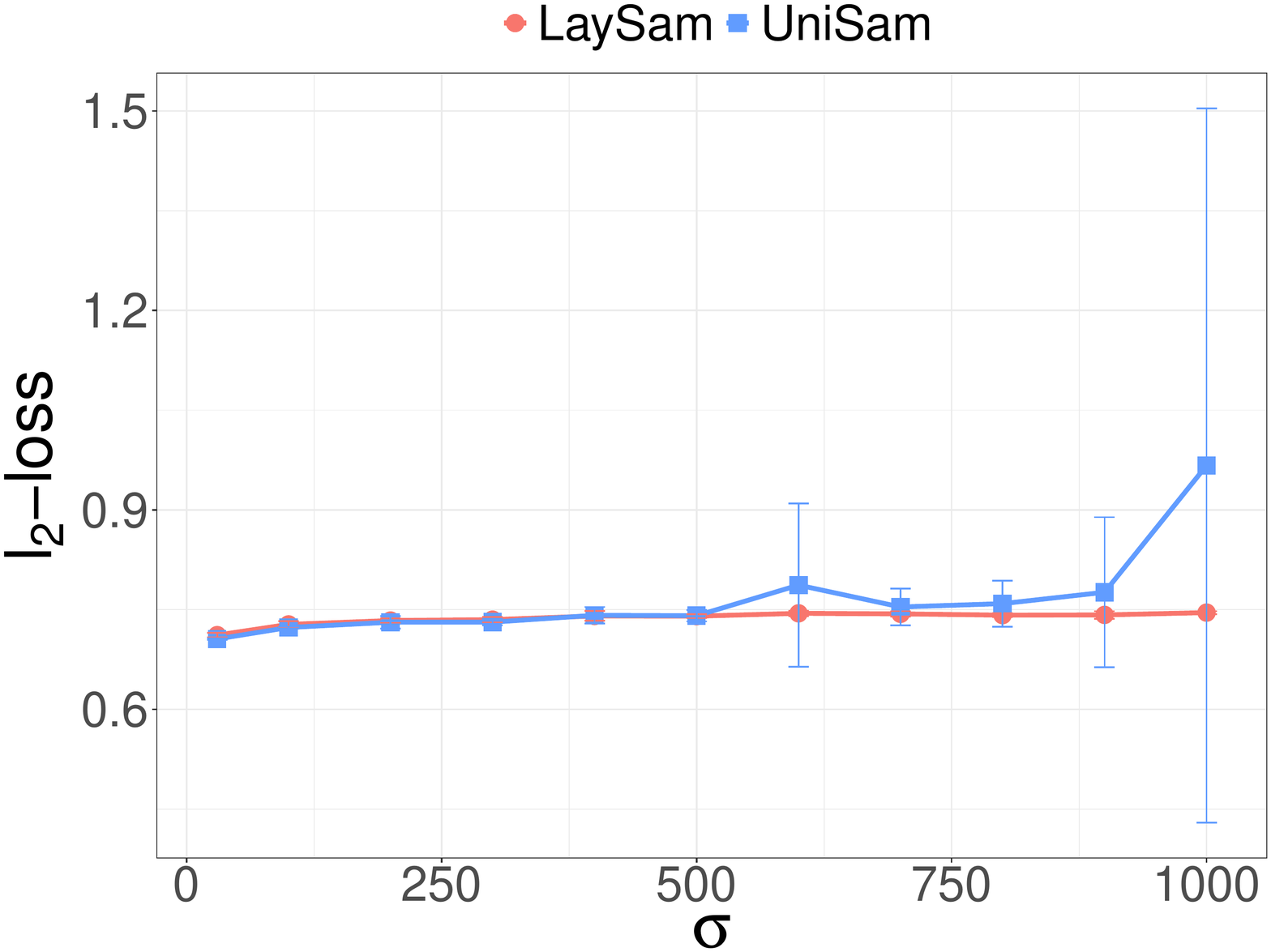}\label{error-lr3}}
	\caption{Linear Regrssion: \texttt{$ \ell_2 $-loss} and stability {\em w.r.t.} $ \sigma $.}
		\end{center}
\end{figure*}

\begin{table*}
\centering
	\subfigure[\texttt{syncluster-gauss-$ \sigma $}, with $ d=20 $, $ k=10$ ,$z=2\cdot10^4$ and coreset size$ =5\cdot 10^4 $.]{
	\begin{tabular}{|l|c|c|c|c|c|c|c|c|c|}
		\hline
		$\sigma$ & \multicolumn{3}{c|}{20} & \multicolumn{3}{c|}{100} & \multicolumn{3}{c|}{200} \\ \hline
		Coreset & LaySam & UniSam & Summary & LaySam & UniSam & Summary & LaySam & UniSam & Summary \\ \hline
		\texttt{$ \ell_1 $-loss} & 3.976 & 3.976 & 4.073 & 3.976 & 3.976 & 4.067 & 3.977 & 4.048 & 4.074 \\ \hline
		\texttt{$ \ell_2 $-loss} & 18.01 & 18 & 18.9 & 18.01 & 18.01 & 18.84 & 18.01 & 18.14 & 18.91 \\ \hline
		\texttt{Prec} & 1 & 1 & 1 & 1 & 1 & 1 & 1 & 1 & 1 \\ \hline
		\texttt{Pre-Rec} & 1 & 0.0506 & 1 & 1 & 0.0499 & 1 & 1 & 0.0503 & 1 \\ \hline
	\end{tabular}
	}
	\subfigure[\texttt{covertype-gauss-$ \sigma $}, with $ k=20 $, $z=10^4$ and coreset size$ =3\cdot10^4 $.]{
	\begin{tabular}{|l|c|c|c|c|c|c|c|c|c|}
		\hline
		$\sigma$ & \multicolumn{3}{c|}{20} & \multicolumn{3}{c|}{100} & \multicolumn{3}{c|}{200} \\ \hline
		Coreset & LaySam & UniSam & Summary & LaySam & UniSam & Summary & LaySam & UniSam & Summary \\ \hline
		\texttt{$ \ell_1 $-loss} & 1.781 & 1.781 & 1.868 & 1.779 & 1.785 & 1.871 & 1.786 & 1.799 & 1.853 \\ \hline
		\texttt{$ \ell_2 $-loss} & 3.501 & 3.501 & 3.818 & 3.505 & 3.523 & 3.874 & 3.515 & 3.576 & 3.767 \\ \hline
		\texttt{Prec} & 1 & 1 & 1 & 1 & 1 & 1 & 1 & 1 & 1 \\ \hline
		\texttt{Pre-Rec} & 1 & 0.0524 & 1 & 1 & 0.0511 & 1 & 1 & 0.0547 & 1 \\ \hline
	\end{tabular}}
	\subfigure[\texttt{skin-uniform-$ \sigma $}, with $ k=40 $, $z=3000$ and coreset size$ =10^4 $.]{
		\begin{tabular}{|l|c|c|c|c|c|c|c|c|c|}
			\hline
			$\sigma$ & \multicolumn{3}{c|}{20} & \multicolumn{3}{c|}{100} & \multicolumn{3}{c|}{200} \\ \hline
			Coreset & LaySam & UniSam & Summary & LaySam & UniSam & Summary & LaySam & UniSam & Summary \\ \hline
			\texttt{$ \ell_1 $-loss} & 0.1906 & 0.1906 & 0.1925 & 0.1883 & 0.1902 & 0.1914 & 0.1954 & 0.2032 & 0.1957 \\ \hline
			\texttt{$ \ell_2 $-loss}($\times 10^{-2}$) & 6.46 & 6.531 & 6.579 & 6.449 & 6.849 & 6.565 & 6.471 & 7.173 & 6.565 \\ \hline
			\texttt{Prec} & 0.9993 & 0.9995 & 0.9993 & 1 & 1 & 1 & 1 & 1 & 1 \\ \hline
			\texttt{Pre-Rec} & 0.9997 & 0.0403 & 1 & 1 & 0.0403 & 1 & 1 & 0.0437 & 1 \\ \hline
		\end{tabular}
	}
\caption{Performance on clustering with outliers.}
\label{tab-1}
\end{table*}

\begin{table*}
\centering
\subfigure[\texttt{synregression-gauss-$ \sigma $}, with $ d=20 $, $ z=10^4 $ and coreset size=$ 2\cdot10^4 $.]{
	\begin{tabular}{|l|c|c|c|c|c|c|}
		\hline
		$\sigma$ & \multicolumn{2}{c|}{20} & \multicolumn{2}{c|}{300} & \multicolumn{2}{c|}{600} \\ \hline
		Coreset & LaySam & UniSam & LaySam & UniSam & LaySam & UniSam \\ \hline
		\texttt{$ \ell_1 $-loss} & 0.7996 & 0.7982 & 0.7989 & 0.8007 & 0.7987 & 0.7998 \\ \hline
		\texttt{$ \ell_2 $-loss} & 1.003 & 0.9993 & 1.001 & 1.01 & 1.002 & 1.025 \\ \hline
		\texttt{Prec} & 0.9305 & 0.931 & 0.9953 & 0.9953 & 0.9975 & 0.9978 \\ \hline
		\texttt{Pre-Rec} & 0.9535 & 0.0108 & 0.9968 & 0.0093 & 0.9975 & 0.0098 \\ \hline
	\end{tabular}
	}
\subfigure[\texttt{PM2.5-uniform-$ \sigma $}, with $ z=1000 $ and coreset size=$ 3000 $.]{
\begin{tabular}{|l|c|c|c|c|c|c|}
	\hline
	$\sigma$ & \multicolumn{2}{c|}{20} & \multicolumn{2}{c|}{500} & \multicolumn{2}{c|}{1000} \\ \hline
	Coreset & LaySam & UniSam & LaySam & UniSam & LaySam & UniSam \\ \hline
	\texttt{$ \ell_1 $-loss} & 0.6292 & 0.6244 & 0.637 & 0.6368 & 0.6349 & 0.6791 \\ \hline
	\texttt{$ \ell_2 $-loss} & 0.7179 & 0.7075 & 0.744 & 0.7473 & 0.7456 & 0.8819 \\ \hline
	\texttt{Prec} & 0.9188 & 0.921 & 0.99 & 0.9901 & 0.9918 & 0.9917 \\ \hline
	\texttt{Pre-Rec} & 0.9714 & 0.0732 & 0.999 & 0.0725 & 1 & 0.0734 \\ \hline
\end{tabular}
}
\caption{Performance on linear regression with outliers.}
\label{tab-2}
\end{table*}

\subsection{Performance}

Figure~\ref{error-cluster}-\ref{error-cluster5} show the performances of clustering on \texttt{$ \ell_2 $-loss} with different $ \sigma $s. The results on \texttt{$ \ell_1 $-loss} are very similar  (more results are summarized in Table \ref{tab-1} and \ref{tab-2}). We can see that \texttt{LaySam} outperforms  the other methods in terms of both synthetic and real-world datasets. Moreover, its performance remains quite stable (with small standard deviation) and is also robust when $ \sigma $ increases. \texttt{UniSam} works well when $ \sigma $ is small, but it becomes very instable when $ \sigma $ rises to large. 
Both \texttt{LaySam} and \texttt{UniSam} outperform \texttt{Summary} on most datasets. 

Similar comparison of the performance for linear regression are shown in Figure~\ref{error-lr}-\ref{error-lr3}.

The three coreset methods achieve very close values of \texttt{recall} and \texttt{precision}.  But \texttt{UniSam} has much lower \texttt{pre-recall} than those of \texttt{Summary} and \texttt{LaySam}.

\section{Conclusion}
To reduce the time complexities of existing algorithms for clustering and linear regression with outliers, we propose a new variant of coreset method which can guarantee the quality for any solution in a local range surrounding the given initial solution. 
%In our supplement, we provide the complete experimental results to evaluate this new coreset method in terms of running time and optimization quality. 
In future, it is worth considering to apply our framework to a broader range of robust optimization problems, such as logistic regression with outliers and Gaussian mixture model with outliers.

\newpage

\bibliography{local}
\bibliographystyle{icml2020}

\section{Proof of Claim 1}

Since $S^C_{in}$ is the set of inliers to $C$, there must exist some value $r_S>0$ such that 
\begin{eqnarray}
S^C_{in}=\{p\mid p\in S, \min_{1\leq j\leq k}||p-c_j||\leq r_S\}. 
\end{eqnarray} 
And therefore
\begin{eqnarray}
S^C_{in}\setminus P_H=\{p\mid p\in S\setminus P_H, \min_{1\leq j\leq k}||p-c_j||\leq r_S\}. \label{for-c1-1}
\end{eqnarray} 
Similarly, there exists some value $r_P>0$ such that
\begin{eqnarray}
P^C_{in}\setminus P_H=\{p\mid p\in P\setminus P_H, \min_{1\leq j\leq k}||p-c_j||\leq r_P\}. \label{for-c1-2}
\end{eqnarray} 
Note $S\setminus P_H=P\setminus P_H$. So, if $r_S\leq r_P$, we have $S^C_{in}\setminus P_H\subseteq P^C_{in}\setminus P_H$. Otherwise, $P^C_{in}\setminus P_H\subseteq S^C_{in}\setminus P_H$.

\section{Proof of Claim 2}
Let $h=(h_1, \cdots, h_d)$, and suppose $h'=(h'_1, \cdots, h'_d)$ is $h$'s nearest neighbor in $\mathcal{G}$, {\em i.e.,} $|h'_d-h_d|\leq \frac{\epsilon}{2d}L$ and $|Dh'_j+h'_d-D h_j-h_d|\leq \frac{\epsilon}{2d}L$ for $1\leq j\leq d-1$. Then, 
\begin{eqnarray}
 |h'_j-h_j|&\leq& \frac{1}{D}(\frac{\epsilon}{2d}L+|h'_d-h_d|)\nonumber\\
 &\leq&\frac{\epsilon}{Dd}L
 \end{eqnarray}
for $1\leq j\leq d-1$. For any $p=(x_1, \cdots, x_d)\in\mathcal{R}_D$, 
\begin{eqnarray}
&&|Res(p, h)-Res(p, h')|\nonumber\\
&\leq& \sum^{d-1}_{j=1}|h'_j-h_j|\cdot |x_j|+|h'_d-h_d|\nonumber\\
&\leq& \sum^{d-1}_{j=1}|h'_j-h_j|\cdot D+|h'_d-h_d|\nonumber\\
&\leq& \epsilon L.
\end{eqnarray}

\end{document}